\documentclass[11pt,a4paper]{article}

\usepackage{amsmath,amsfonts,amssymb,amsthm}
\usepackage{mathbbol}
\usepackage{amsthm}
\usepackage{graphicx,color}
\usepackage{boxedminipage}
\usepackage[ruled, linesnumbered, boxed,lined]{algorithm2e}
\usepackage{algorithmicx}
\usepackage{framed}
\usepackage{thmtools}
\usepackage{thm-restate}
\usepackage{xspace}
\usepackage{vmargin}
\setmarginsrb{1in}{1in}{1in}{1in}{0pt}{0pt}{0pt}{6mm}
\usepackage{todonotes}
 \usepackage[pdftex, plainpages = false, pdfpagelabels, 
                 bookmarks=false,
                 bookmarksopen = true,
                 bookmarksnumbered = true,
                 breaklinks = true,
                 linktocpage,
                 pagebackref,
                 colorlinks = true,  
                 linkcolor = blue,
                 urlcolor  = blue,
                 citecolor = red,
                 anchorcolor = green,
                 hyperindex = true,
                 hyperfigures
                 ]{hyperref} 
 \usepackage{xifthen}
 \usepackage{tabularx}
\usepackage{arydshln}
 \usetikzlibrary{calc}
 \usepackage{enumitem}
 \usepackage{tcolorbox} 
 
 \usepackage{mathtools}

\newcommand{\Oh}{\mathcal{O}}

\newcommand{\cost}{{\sf cost}}

\newcommand{\opt}{{\sf Opt}}

\DeclareMathOperator{\operatorClassNP}{{\sf NP}}
\newcommand{\classNP}{\ensuremath{\operatorClassNP}}
\DeclareMathOperator{\operatorClassCoNP}{{\sf coNP}}
\DeclareMathOperator{\operatorClassPoly}{{\sf poly}}

\newcommand{\nopoly}{\ensuremath{\operatorClassNP\not\subseteq\operatorClassCoNP/\operatorClassPoly}}
\newcommand{\nonopoly}{\ensuremath{\operatorClassNP\subseteq\operatorClassCoNP/\operatorClassPoly}}
\DeclareMathOperator{\operatorClassFPT}{{\sf FPT}\xspace}
\newcommand{\classFPT}{\ensuremath{\operatorClassFPT}\xspace}

\newcommand{\bfC}{\mathbf{C}}

\newcommand{\bfc}{\mathbf{c}}

\newcommand{\bff}{\mathbf{f}} 
\newcommand{\bfS}{\mathbf{S}} 
 
\newcommand{\bfx}{\mathbf{x}} 
\newcommand{\bfX}{\mathbf{X}} 
\newcommand{\bfY}{\mathbf{Y}} 
\newcommand{\bfZ}{\mathbf{Z}} 
\newcommand{\bfy}{\mathbf{y}} 
\newcommand{\bfv}{\mathbf{v}} 

\newtheorem{theorem}{Theorem}
\newtheorem{lemma}{Lemma}
\newtheorem{claim}{Claim}[section]
\newtheorem{corollary}{Corollary}

\newtheorem{observation}{Observation}
\newtheorem{proposition}{Proposition}
\newtheorem{redrule}{Reduction Rule}

\newcommand{\pname}{\textsc}
\newcommand{\ProblemFormat}[1]{\pname{#1}}
\newcommand{\ProblemIndex}[1]{\index{problem!\ProblemFormat{#1}}}
\newcommand{\ProblemName}[1]{\ProblemFormat{#1}\ProblemIndex{#1}{}\xspace}

\newcommand{\probEClust}{\ProblemName{Decision Equal Clustering}}
\newcommand{\probMEClust}{\ProblemName{Equal Clustering}}
\newcommand{\probPEClust}{\ProblemName{Parameterized Equal Clustering}}
\newcommand{\probClust}{\ProblemName{$k$-Median}}

\newcommand{\probPM}{\ProblemName{Minimum Weight Perfect Matching}}

\newcommand{\probDM}{\ProblemName{Perfect $r$-Set Matching}}



\newlength{\RoundedBoxWidth}
\newsavebox{\GrayRoundedBox}
\newenvironment{GrayBox}[1]%
   {\setlength{\RoundedBoxWidth}{.93\textwidth}
    \def\boxheading{#1}
    \begin{lrbox}{\GrayRoundedBox}
       \begin{minipage}{\RoundedBoxWidth}}%
   {   \end{minipage}
    \end{lrbox}
    \begin{center}
    \begin{tikzpicture}%
       \node(Text)[draw=black!20,fill=white,rounded corners,%
             inner sep=2ex,text width=\RoundedBoxWidth]%
             {\usebox{\GrayRoundedBox}};
        \coordinate(x) at (current bounding box.north west);
        \node [draw=white,rectangle,inner sep=3pt,anchor=north west,fill=white] 
        at ($(x)+(6pt,.75em)$) {\boxheading};
    \end{tikzpicture}
    \end{center}}     

\newenvironment{defproblemx}[2][]{\noindent\ignorespaces%
                                \FrameSep=6pt%
                                \parindent=0pt%
                \vspace*{-1.5em}
                \ifthenelse{\isempty{#1}}{%
                  \begin{GrayBox}{\textsc{#2}}%
                }{%
                  \begin{GrayBox}{\textsc{#2} parameterized by~{#1}}%
                }
                \begin{tabular*}{\textwidth}{@{\hspace{.1em}} >{\itshape} p{1.8cm} p{0.8\textwidth} @{}}%
            }{
                \end{tabular*}%
                \end{GrayBox}%
                \ignorespacesafterend
            }

\newcommand{\defproblema}[3]{
  \begin{defproblemx}{#1}
    Input:  & #2 \\
    Task: & #3
  \end{defproblemx}
}%

\begin{document}

\title{Lossy Kernelization of Same-Size Clustering\thanks{The research leading to these results have  been supported by the Research Council of Norway via the project ``MULTIVAL" (grant no. 263317) and the European Research Council (ERC) via grant LOPPRE, reference 819416.}}

\author{
Sayan Bandyapadhyay\thanks{Department of Informatics, University of Bergen, Bergen, Norway}\addtocounter{footnote}{-1}
\and
Fedor V. Fomin\thanks{\{Sayan.Bandyapadhyay, Fedor.Fomin, Petr.Golovach, Nidhi.Purohit, Kirill.Simonov\}@uib.no}\addtocounter{footnote}{-1}
\and
Petr A. Golovach\footnotemark{} \addtocounter{footnote}{-1}
\and
Nidhi Purohit\footnotemark{} \addtocounter{footnote}{-1}
\and 
Kirill Simonov\footnotemark
}

\date{}

\maketitle

\begin{abstract}
In this work, we study the $k$-median clustering problem with an additional equal-size constraint on the clusters, from the perspective of parameterized preprocessing. Our main result is the first lossy ($2$-approximate) polynomial kernel for this problem, parameterized by the cost of clustering. We complement this result by establishing lower bounds for the problem that eliminate the existences of an (exact) kernel of polynomial size and a PTAS.  
\end{abstract}

\section{Introduction}\label{sec:intro}

\emph{Lossy kernelization} stems from \emph{parameterized complexity}, a branch in theoretical computer science that studies complexity of problems as functions of multiple \emph{parameters} of the input or output~\cite{DowneyF13}. 
A central notion in parameterized complexity is \emph{kernelization}, which is a generic technique for designing efficient algorithms availing a polynomial time preprocessing step that transforms a ``large'' instance of a problem into a  smaller, equivalent instance. Naturally, the preprocessing step is called the \emph{kernelization algorithm} and the smaller instance is called the \emph{kernel}. One limitation of the classical kernelization technique is that kernels can only analyze ``lossless'' preprocessing, in the sense that a kernel must be equivalent to the original instance. This is why most of the interesting models of problems arising from machine learning, e.g.,  clustering, are intractable from the perspective of kernelization. Lossy or approximate kernelization is a successful attempt of combining kernelization   with approximation algorithms. Informally, in lossy kernelization, given an instance of the problem and a parameter, we would like the kernelization algorithm to output a reduced instance of size polynomial in the parameter; however  the notion of equivalence is relaxed in the following way. Given a $c$-approximate solution (i.e., one with the cost within $c$-factor of the optimal cost) to the reduced instance, it should be possible to produce in polynomial time an 
$\alpha c$-approximate solution to the original instance. The factor $\alpha$ is the loss   incurred while going from reduced instance to the original instance.  
The notion of lossy kernelization was introduced by Lokshtanov et al. in  \cite{DBLP:conf/stoc/LokshtanovPRS17}.
 However,   most of the developments of lossy kernelization up to now are in graph algorithms 
  \cite{DBLP:journals/corr/EibenKMP17,DBLP:journals/corr/Siebertz17a,DBLP:conf/fsttcs/KrithikaM0T16,DBLP:journals/corr/abs-1708-00622,DBLP:journals/corr/newlossy}, see also \cite[Chapter~23]{FominLSZ19} for an overview.  
  
One of the actively developing areas of parameterized complexity concerns \emph{fixed-parameter tractable}- or FPT-approximation. We refer to the survey 
 \cite{DBLP:journals/algorithms/FeldmannSLM20} for a nice overview of the area. Several important advances on FPT-approximation concern clustering problems. It includes    
tight algorithmic and complexity results for $k$-means and $k$-median   \cite{Cohen-AddadG0LL19} and constant factor FPT-approximation for capacitated clustering \cite{Cohen-AddadL19}.  The popular approach for   data compression used for FPT-approximation of clustering is based on \emph{coresets}.  
The notion of coresets originated from computational geometry. It was introduced by 
Har-Peled and Mazumdar \cite{har2004coresets} for $k$-means and  $k$-median clustering. Informally, a coreset is  a summary of the data that for every set   of $k$ centers,  approximately (within $(1\pm \epsilon)$ factor) preserves the optimal clustering cost.

 Lossy kernels and coresets have a lot of similarities. Both compress the space compared to the original data, and  any  algorithm can be applied on a coreset or kernel to efficiently retrieve a solution with guarantee almost the same as the one provided by the algorithm on the original input. 
The crucial difference is that  coreset constructions result  in a small set of weighted points.   The weights could be up to the input size $n$. Thus a coreset of size polynomial in  $k/\epsilon$,  is not a polynomial sized lossy kernel for parameters $k, \epsilon$ because of the $\log{n}$ bits required to encode the weights. Moreover,  usually   coreset constructions do not bound the number of coordinates or dimension of the points. 

 While the notion of lossy kernelization proved to be useful in the design of graph algorithms, we are not aware of its applicability in clustering. This brings us to the following question.
    \begin{tcolorbox}[colback=green!5!white,colframe=blue!40!black]
What can lossy kernelization offer to clustering?
 \end{tcolorbox}
  
In this work, we make the first step towards the development  of lossy kernels for clustering problems.     
Our main result   is the design of a lossy kernel for a  variant of the ubiquitous \probClust   clustering with   clusters of equal sizes. More precisely, consider 
 a collection (multiset) of points from $\mathbb{Z}^d$ with $\ell_p$-norm.  Thus every point is a $d$-dimensional vector with integer coordinates. 
 For a nonnegative integer $p$, we use $\|\bfx\|_p$ to denote the $\ell_p$-norm of a $d$-dimensional vector $\bfx=(x[1],\ldots,x[d])\in\mathbb{R}^d$, that is, for $p\geq 1$, 
$$\|\bfx\|_p=\big(\sum_{i=1}^d|x[i]|^p\big)^{1/p}$$
 and for $p=0$, $\|\bfx\|_0$ is the number of nonzero elements of $\bfx$, i.e., the Hamming norm. For any subset of points $T\subseteq \mathbb{Z}^d$, we define 
\begin{equation*}
\cost_p(T)=\min_{\bfc\in\mathbb{R}^d}\sum_{\bfx\in T}\|\bfc-\bfx\|_p.
\end{equation*}
Then \probClust\footnote{Traditionally this problem is studied with real input points, but because of the choice of the parameterization, it is natural for us to assume that points have integer coordinates. As the coordinates can be scaled, this does not lead to the loss of generality.} clustering (without constraints) is the task of  finding  a partition $\{X_1,\ldots,X_k\}$  of a given set $\bfX\subseteq \mathbb{Z}^d$ of points minimizing the sum
\[
\sum_{i=1}^k  \cost_p(X_i).
\]

In many real-life scenarios, it is desirable to cluster data into clusters of equal sizes. For example,
to tailor teaching methods to the specific needs of various students, one would be interested in allocating  $k$ fair class sizes by grouping students with homogeneous abilities and skills \cite{HoppnerK08}. In scheduling, the standard task is to distribute $n$ jobs to $k$ machines while keeping identical workloads on each machine and simultaneously reducing the configuration time. In the setting of designing a conference program, one might be interested in allocating $n$ scientific papers according to their similarities to $k$ ``balanced'' sessions
\cite{DBLP:conf/iccS/Vallejo-HuangaM17}.

 The following model is an attempt to capture such scenarios.
  
\defproblema{\probMEClust}%
{A collection (multiset) $\bfX=\{\bfx_1,\ldots,\bfx_n\}$ of $n$ points of $\mathbb{Z}^d$ and a positive integer $k$ such that $n$ is divisible by $k$.}%
{Find a partition $\{X_1,\ldots,X_k\}$ ($k$-clustering) of $\bfX$ with $|X_1|=\cdots=|X_k|=\frac{n}{k}$ minimizing 
\[
\sum_{i=1}^k  \cost_p(X_i).
\]
 }
First, note that \probMEClust is a restricted variant of the capacitated version \cite{Cohen-AddadL19} of \probClust where the size of each cluster is required to be bounded by a given number $U$. Also note, that some points in $\bfX$ may be identical. (In the above examples, several students, jobs, or scientific papers can have identical features but could be assigned to different clusters due to the size limitations.) We refer to the multisets $X_1,\ldots,X_k$ as the  \emph{clusters}.

To describe the lossy-kernel result, we need to define the parameterized version of \probMEClust with the cost of clustering $B$ (the budget) being the parameter. Following the framework of parameterized kernelization \cite{DBLP:conf/stoc/LokshtanovPRS17}, when the cost of  an optimal clustering exceeds the budget, we assume it is equal to $B+1$. 
More precisely, in 
 \probPEClust,   we  are given an additional integer $B$ (budget parameter). The task is to find a 
$k$-clustering $\{X_1,\ldots,X_k\}$ with $|X_1|=\cdots =|X_k|$ and minimizing the value
\begin{equation*}
\cost_p^B(X_1,\ldots,X_k)=
\begin{cases}
\sum_{i=1}^k  \cost_p(X_i)  &\mbox{if } \sum_{i=1}^k  \cost_p(X_i) \leq B,\\
B+1&\mbox{otherwise.}
\end{cases}
\end{equation*}
\medskip

Our first main result is the following theorem providing a polynomial $2$-approximate kernel. 
\begin{restatable}{theorem}{Lossy}{\label{thm:lossykernel1}}
For every nonnegative integer constant $p$, \probPEClust admits a $2$-approximate kernel when parameterized by $B$, where the output collection of points has $\mathcal{O}(B^2)$ points of $\mathbb{Z}^{d'}$ with $d'= \mathcal{O}(B^{p+2})$, where each coordinate of a point takes an absolute value of $\mathcal{O}(B^3)$.
\end{restatable}

In other words, the theorem provides a polynomial time algorithm that compresses the original instance $\bfX$ to a new instance 
whose size is bounded by a polynomial of $B$ and such that any $c$-approximate solution in the new instance can be turned in polynomial time to a $2c$-approximate solution of the original instance. 

A natural question is whether  the approximation ratio of lossy kernel in Theorem~\ref{thm:lossykernel1} is optimal.  While we do not have a complete answer to this question, we provide lower bounds supporting our study of the problem from the perspective of approximate kernelization. Our next result rules out the existence of an ``exact'' kernel for the problem.  To state the result, we need to define the decision version of  \probMEClust.  In this version, we call it \probEClust, the question is whether for a given budget $B$, there is 
a $k$-clustering $\{X_1,\ldots,X_k\}$ with clusters of the same size such that $\sum_{1\leq i \leq k}  \cost_p(X_i) \leq B$. 

\begin{restatable}{theorem}{Nokernel}{\label{thm:no-kern}}   
For $\ell_0$ and $\ell_1$-norms,
\probEClust has no polynomial kernel when parameterized by $B$, unless $\nonopoly$, even if the input points are binary, that is, are from $\{0,1\}^d$. 
\end{restatable}

On the other hand, we prove that  \probEClust admits a polynomial kernel when parameterized by $k$ and $B$.
 
\begin{restatable}{theorem}{Kernel}\label{thm:kern}
For every nonnegative integer constant $p$,
\probEClust admits a polynomial kernel when parameterized by $k$ and $B$, where the output collection of points has $\mathcal{O}(kB)$ points of $\mathbb{Z}^{d'}$ with $d'= \mathcal{O}(kB^{p+1})$ and  each coordinate of a point takes an absolute value of $\mathcal{O}(kB^2)$.
\end{restatable}

 When it comes to approximation in polynomial time, we show (Theorem~\ref{thm:no-PTAS}) that   it is 
 \classNP-hard to obtain a $(1+\epsilon_c)$-approximation for \probMEClust with $\ell_0$ $($or $\ell_1)$ distances for some $\epsilon_c > 0$. However, parameterized by $k$ and $\epsilon$, the standard techniques yield $(1+\epsilon)$-approximation in FPT time.
 For $\ell_2$ norm, there is a general framework for designing algorithms of this form for \probClust with additional constraints on cluster sizes, introduced by Ding and Xu~\cite{ding2020unified}. The best-known improvements by Bhattacharya et al.~\cite{Bhattacharya2018} achieve a running time of $2^{\widetilde{\mathcal{O}}(k/\epsilon^{\mathcal{O}(1)})} n^{\mathcal{O}(1)}d$ in the case of \probMEClust, where $\widetilde{\mathcal{O}}$ hides polylogarithmic factors. In another line of work, FPT-time approximation is achieved via constructing small-sized coresets of the input, and the work~\cite{DBLP:journals/corr/abs-2007-10137} guarantees an $\epsilon$-coreset for \probMEClust (in $\ell_2$ norm) of size $(k d\log n/\epsilon )^{\mathcal{O}(1)}$, and consequently a $(1 + \epsilon)$-approximation  algorithm with running time $2^{\widetilde{\mathcal{O}}(k/\epsilon^{\mathcal{O}(1)})} (nd)^{\mathcal{O}(1)}$.

 Moreover, specifically for \probMEClust, simple $(1 + \epsilon)$-approximations with similar running time can be designed directly via sampling. A seminal work of Kumar et al.~\cite{DBLP:journals/jacm/KumarSS10} achieves a $(1 + \epsilon)$-approximation for \probClust (in $\ell_2$ norm)  with running time $2^{\widetilde{\mathcal{O}}(k/\epsilon^{\mathcal{O}(1)})} nd$. The algorithm proceeds as follows. First, take a small uniform sample of the input points, and by guessing ensure that the sample is taken only from the largest cluster. Second, estimate the optimal center of this cluster from the sample. In the case of \probClust, Theorem 5.4 of~\cite{DBLP:journals/jacm/KumarSS10} guarantees that from a sample of size $(1/\epsilon)^{\mathcal{O}(1)}$ one can compute in time $2^{(1/\epsilon)^{\mathcal{O}(1)}} d$ a set of candidate centers such that at least one of them provides a $(1 + \epsilon)$-approximation to the cost of the cluster. Finally, ``prune'' the set of points so that the next largest cluster is at least $\Omega(1/k)$ fraction of the remaining points and continue the same process with one  less cluster. One can observe that in the case of \probMEClust, a simplification of the above algorithm suffices: one does not need to perform the ``pruning'' step, as we are only interested in clusterings where all the clusters have size exactly $n / k$.
 Thus, $(1/\epsilon)^{\mathcal{O}(1)}$-sized uniform samples from each of the clusters can be computed immediately in total time $2^{\widetilde{\mathcal{O}}(k/\epsilon^{\mathcal{O}(1)})} nd$. This achieves $(1 + \epsilon)$-approximation for \probMEClust with the same running time as the algorithm of Kumar et al. In fact, the same procedure works for $\ell_0$ norm as well, where for estimating the cluster center it suffices to compute the optimal center of a sample of size $\mathcal{O}(1/\epsilon^2)$, as proven by Alon and Sudakov~\cite{AlonS99}. Thus, in terms of FPT approximation, \probMEClust is surprisingly ``simpler'' than its unconstrained variant \probClust, however, our hardness result of Theorem~\ref{thm:no-PTAS} shows that the problems are similarly hard in terms of polynomial time approximation.

 \medskip\noindent
\textbf{Related work.} Since the work of 
Har-Peled and Mazumdar \cite{har2004coresets} for $k$-means and  $k$-median clustering, designing small coresets for clustering has become a flourishing research direction.
 For these problems, after a series of interesting works,  the best-known upper bound on coreset size in general metric space is  
$\Oh((k\log n)/\epsilon^2)$
 \cite{feldman2011unified} and the lower bound is known to be $\Omega(({k}\log n)/{\epsilon})$ \cite{baker2019coresets}. For the Euclidean space (i.e., $\ell_2$-norm) of dimension $d$, it is possible to construct coresets of size $(k/\epsilon)^{\Oh(1)}$ \cite{feldman2013turning,sohler2018strong}. Remarkably,  the size of the coresets in this case does not depend on $n$ and $d$.  For \probMEClust, the best known coreset size of $(kd\log n/\epsilon)^{O(1)}$ (for $p=2$) follows from coresets for the more general  capacitated clustering problem \cite{Cohen-AddadL19,DBLP:journals/corr/abs-2007-10137}.

Clustering is undoubtedly one of the most common procedures in unsupervised  machine learning.  We refer to the   book \cite{Aggarwalbook13} for an overview on clustering. \probMEClust belongs to a wide class of clustering with constraints on the sizes of the clusters.  
In many applications of clustering,  constraints come naturally \cite{basu2008constrained}.  In particular, there is a rich literature on approximation algorithms for various versions of capacitated clustering.  

While  for the capacitated version of $k$-median and $k$-means in general metric space,  no polynomial time $O(1)$-approximation is known,  bicriteria constant-approximations   violating either the capacity constraints or the constraint on the number of clusters, by an $O(1)$ factor  can be obtained~\cite{ByrkaRU16,ByrkaFRS15,CharikarGTS02,ChuzhoyR05,DemirciL16,Li15,Li17}.  Cohen-Addad and Li~\cite{Cohen-AddadL19} designed FPT $\approx 3$- and $\approx 9$-approximation with parameter $k$ for the capacitated version of $k$-median and $k$-means, respectively. For these problems in the Euclidean plane, Cohen-Addad~\cite{Cohen-Addad20} obtained a true PTAS. Moreover, for higher dimensional spaces (i.e., $d\ge 3$), he designed a $(1+\epsilon)$-approximation that runs in time $n^{{(\log n/\epsilon)}^{O(d)}}$~\cite{Cohen-Addad20}. Being a restricted version of capacitated clustering, \probMEClust admits all the approximation results mentioned above.  

 \medskip\noindent
\textbf{Our approach.} We briefly sketch the main ideas behind the construction of our lossy kernel for \probPEClust. 
The lossy kernel's main ingredients are a) a polynomial algorithm based on an algorithm for computing a minimum weight perfect matching in bipartite graphs, b) preprocessing rules reducing the size and dimension of the problem, and c) a greedy algorithm. Each of the steps is relatively simple and easily implementable. However, the proof that these steps result in the lossy kernel with required properties is not easy. 

Recall that for a given budget $B$, we are looking for a $k$-clustering of a collection of point $\bfX=\{\bfx_1,\ldots,\bfx_n\}$ into $k$ clusters of the same size minimizing the cost. We also assume that the cost is $B+1$ if the instance points do not admit a clustering of cost at most $B$. Informally, we are only interested in optimal clustering when its cost does not exceed the budget. 
 First,  if the cluster's size $s=\frac{n}{k}$ is sufficiently large (with respect to the budget), we can construct an optimal clustering in polynomial time. More precisely, we prove that if $s\geq 4B+1$, then the clusters' medians could be selected from $\bfX$.  Moreover, we show how to identify the (potential) medians in polynomial time. In this case, constructing an optimal  $k$-clustering could be reduced to the classical problem of computing a perfect matching of minimum weight in a bipartite graph. 

The case of cluster's size $s\leq 4B$ is different. We apply a set of reduction rules. These rules run in polynomial time. After exhaustive applications of reduction rules, we either correctly conclude that the considered instance has no clustering of cost at most $B$ or constructs an equivalent reduced instance.  In the equivalent instance, the dimension is reduced to $\Oh(kB^{p+1})$ while the absolute values of the coordinates of the points are in $\Oh(kB^2)$. 

Finally, we apply the only approximate reduction on the reduced instance. The approximation procedure is greedy: whenever there are $s$ identical points, we form a cluster out of them. For the points remaining after the exhaustive application of the greedy procedure, we conclude that either there is no clustering of cost at most $B$ or the number of points is $\Oh(B^2)$.  This construction leads us to the lossy kernel. However the greedy selection of the clusters composed of identical points maybe is not optimal. In particular, the reductions used to obtain our algorithmic lower bounds given in Sections~\ref{sec:kern} and \ref{sec:apxhard} exploit the property that it may be beneficial to split a block of $s$ identical points between distinct clusters.

 Nevertheless,  the greedy clustering of identical points leads to a $2$-approximation. The proof of this fact requires some work. 
We evaluate the clustering cost obtained from a given optimal clustering by swapping some points to form clusters composed of identical points. Further, we upper bound the obtained value by the cost of the optimum clustering. For the last step, we introduce an auxiliary clustering problem formulated as a min-cost flow   problem. This reduction allows to evaluate the cost and obtain the required upper bound.

\medskip\noindent
\textbf{Organization of the paper.} The remaining part of the paper is organized as follows. In Section~\ref{sec:prelim}, we introduce basic notation and show some properties of clusterings. In Section~\ref{sec:kernelization}, we show our main result that   \probPEClust admits a lossy kernel. In Section~\ref{sec:kern}, we complement this result by proving that it is unlikely that \probEClust admits an (exact) kernel of polynomial size when parameterized by $B$. Still, the problem has a polynomial kernel when parameterized by $B$
 and $k$. In Section~\ref{sec:apxhard}, we show that \probMEClust is \textsf{APX}-hard. We conclude in Section~\ref{sec:concl} by stating some open problems.

\section{Preliminaries}\label{sec:prelim}
In this section, we give basic definition and introduce notation used throughout the paper. We also state some useful auxiliary results.

\paragraph{Parameterized complexity and kernelization.} We refer to the recent books~\cite{CyganFKLMPPS15,FominLSZ19} for the formal introduction to the area. Here we only define the notions used in our paper. 

Formally, a \emph{parameterized problem} $\Pi$ is a subset of $\Sigma^*\times \mathbb{N}$, where $\Sigma$ is a finite alphabet. Thus, an instance of $\Pi$ is a pair $(I,k)$, where $I\subseteq\Sigma^*$ and $k$ is a nonnegative integer called a \emph{parameter}. It is said that a parameterized problem $\Pi$ is \emph{fixed-parameter tractable} (\classFPT) if it can be solved in $f(k)\cdot |I|^{\Oh(1)}$ time for some computable function $f(\cdot)$. 

A \emph{kernelization} algorithm (or \emph{kernel}) for a parameterized problem $\Pi$ is an algorithm that, given an instance $(I,k)$ of $\Pi$, in polynomial time produces an instance $(I',k')$ of $\Pi$ such that
\begin{itemize}
\item[(i)] $(I,k)\in \Pi$ if and only if $(I',k')\in \Pi$, and
\item[(ii)] $|I'|+k'\leq g(k)$ for a computable function $g(\cdot)$.
\end{itemize}
The function $g(\cdot)$ is called the \emph{size} of a kernel; a kernel is \emph{polynomial} if $g(\cdot)$ is a polynomial. 
Every decidable \classFPT problem admits a kernel. However, it is unlikely that all \classFPT problems have polynomial kernels and the parameterized complexity theory provide tools for refuting the existence of polynomial kernels up to some reasonable complexity assumptions. The standard assumption here is that  $\nopoly$. 

We also consider the parameterized analog of optimization problems. Since we only deal with minimization problems where the minimized value is nonnegative, we state the definitions only for optimization problems of this type. 
A \emph{parameterized minimization} problem $P$ is a computable function 
\begin{equation*}
P\colon \Sigma^*\times \mathbb{N} \times \Sigma^* \rightarrow \mathbb{R}_{\geq 0} \cup \{+\infty\}.
\end{equation*}
An \emph{instance} of $P$ is a pair $(I,k)\in\Sigma^*\times \mathbb{N}$, and a \emph{solution} to $(I,k)$ is a string $s\in\Sigma^*$ such that $|s|\leq |I|+k$. 
The instances of a parameterized minimization problem $P$ are pairs $(I,k) \in \Sigma^* \times \mathbb{N}$, and a solution to $(I,k)$ is simply a string $s \in \Sigma^*$, such that $|s| \leq |I|+k$. Then the function $P(\cdot,\cdot,\cdot)$ defines the \emph{value} $P(I,k,s)$ of a solution $s$ to an instance $(I,k)$. The optimum value of an instance $(I,k)$ is 
\begin{equation*} 
 \opt_{P}(I,k)=\min_{s \in \Sigma^* \text{ s.t. } |s| \leq |I|+k}P(I,k,s).
\end{equation*}
A solution $s$ is \emph{optimal} if $\opt_P(I,k)=P(I,k,s)$. A parameterized minimization problem $P$ is said to be \classFPT if there is an algorithm that for each instance $(I,k)$ of $P$ computes an optimal solution $s$ in $f(k)\cdot |I|^{\Oh(1)}$ time, where $f(\cdot)$ is a computable function. Let $\alpha\geq 1$ be a real number. An \classFPT $\alpha$-approximation algorithm for $P$ is an algorithm that in $f(k)\cdot |I|^{\Oh(1)}$ time computes a solution $s$ for $(I,k)$ such that $P(I,k,s)\leq \alpha\cdot\opt_P(I,k)$, where $f(\cdot)$ is a computable function. 

It is useful for us to make some comments about defining $P(\cdot,\cdot,\cdot)$ for the case when the considered problem is parameterized by the solution value. For simplicity, we do it informally and refer to~\cite{FominLSZ19}  for details and explanations. If $s$ is not a ``feasible'' solution to an instance $(I,k)$, then it is convenient to assume that $P(I,k,s)=+\infty$. Otherwise, if $s$ is ``feasible''  but its value is at least $k+1$, we set $P(I,k,s)=k+1$. 

\paragraph{Lossy  kernel.} 
Finally  we define \emph{$\alpha$-approximate} or \emph{lossy} kernels for parameterized minimization problems. Informally, an $\alpha$-approximate kernel of size $g(\cdot)$ is a polynomial time algorithm, that given an instance $(I,k)$, outputs an instance $(I',k')$ such that $|I|+k \leq g(k)$ and any $c$-approximate solution $s'$ to $(I',k')$ can be turned in polynomial time into a $(c \cdot \alpha)$-approximate solution $s$ to the original instance $(I,k)$.  More precisely, let $P$ be a parameterized minimization problem and let $\alpha\geq 1$. An \emph{$\alpha$-approximate} (or \emph{lossy}) kernel for $P$ is a pair of polynomial algorithms $\mathcal{A}$ and $\mathcal{A}'$ such that 
\begin{itemize}  
\item[(i)] given an instance $(I,k)$, $\mathcal{A}$ (called a \emph{reduction algorithm}) computes an instance $(I',k')$ with $|I'|+k'\leq g(k)$, where $g(\cdot)$ is a computable function,
\item[(ii)] the algorithm $\mathcal{A}'$ (called a \emph{solution-lifting algorithm}), given the initial instance $(I,k)$, the instance $(I',k')$ produced by $\mathcal{A}$, and a solution $s'$ to $(I',k')$, computes an solution $s$ to $(I,k)$ such that
\begin{equation*}
\frac{P(I,k,s)}{\opt_P(I,k)}\leq \alpha\cdot \frac{P(I',k',s')}{\opt_P(I',k')}.
\end{equation*} 
\end{itemize}
To simplify notation, we assume here that $\frac{P(I,k,s)}{\opt_P(I,k)}=1$ if $\opt_P(I,k)=0$ and use the same assumption   
for $\frac{P(I',k',s')}{\opt_P(I',k')}$.
As with classical kernels, $g(\cdot)$ is called the \emph{size} of an approximate kernel, and an approximate kernel is polynomial if $g(\cdot)$ is a polynomial.

\paragraph{Vectors and clusters.} For a vector $\bfx\in\mathbb{R}^d$, we use $\bfx[i]$ to denote the $i$-th element of the vector for $i\in\{1,\ldots,d\}$. 
For a set of indices  $R \subseteq \{1,\ldots,d\}$, $\bfx[R]$ denotes the vector of $\mathbb{R}^{|R|}$ composed by the elements of vector $\bfx$ from set $R$, that is,
if $R=\{i_1,\ldots,i_r\}$ with $i_1<\ldots<i_r$ and $\bfy=\bfx[R]$, then $\bfy[j]=\bfx[i_j]$ for $j\in\{1,\ldots,r\}$. In our paper, we consider collections $\bfX$  of points of $\mathbb{Z}^d$. We underline that  some points of such a collection may be identical. 
However, to simplify notation, we assume throughout the paper that the identical points of  $\bfX$ are distinct elements of $\bfX$ assuming that the points are supplied with unique identifiers. By this convention, we often refer to (sub)collections of points as (sub)sets and apply the standard set notation.  

Let $X$ be a collection of points of $\mathbb{Z}^d$. 
For a vector $\bfc\in\mathbb{R}^d$, we define the \emph{cost of $X$ with respect to $\bfc$} as
\begin{equation*}
\cost_p(X,\bfc)=\sum_{\bfx\in X}\|\bfc-\bfx\|_p.
\end{equation*}
Slightly abusing notation we often refer to $\bfc$ as a (given) \emph{median} of $X$. We say that $\bfc^*\in\mathbb{R}^d$ is an \emph{optimum median} of $X$ if 
$\cost_p(X)=\cost_p(X,\bfc^*)=\min_{\bfc\in \mathbb{R}^d}\cost_p(X,\bfc)$. Notice that the considered collections of points have integer coordinates but the coordinates of medians are not constrained to integers and may be real. 

Let $\bfX=\{\bfx_1,\ldots,\bfx_n\}$ a collection of points of $\mathbb{Z}^d$ and let $k$ be a positive integer such that $n$ is divisible by $k$. We say that a partition $\{X_1,\ldots,X_k\}$ of $\bfX$ is an 
\emph{equal $k$-clustering} of $\bfX$ if  $|X_i|=\frac{n}{k}$ for all $i \in \{1,\ldots,k\}$. For an equal $k$-clustering  $\{X_1,\ldots,X_k\}$  and given vectors $\bfc_1,\ldots,\bfc_k$, we define the \emph{cost of clustering with respect to $\bfc_1,\ldots,\bfc_k$} as 
\begin{equation*}
 \cost_p(X_1,\ldots,X_k,\bfc_1,\ldots,\bfc_k)=\sum_{i=1}^k\cost_p(X_i,\bfc_i).
 \end{equation*} 
 The \emph{cost} of an equal $k$-clustering  $\{X_1,\ldots,X_k\}$  is
 $\cost_p(X_1,\ldots,X_k)=\cost_p(X_1,\ldots,X_k,\bfc_1,\ldots,\bfc_k)$, where $\bfc_1,\ldots,\bfc_k$ are optimum medians of $X_1,\ldots,X_k$, respectively.
For an integer $B\geq 0$,
 \begin{equation*}
 \cost_p^B(X_1,\ldots,X_k)=
 \begin{cases}
\cost_p(X_1,\ldots,X_k)&\mbox{if } \cost_p(X_1,\ldots,X_k)\leq B,\\
B+1&\mbox{otherwise.}
\end{cases}
 \end{equation*}  
We define
\begin{equation*}
\opt(\bfX,k)=\min\{\cost_p(X_1,\ldots,X_k)\mid \{X_1,\ldots,X_k\}\text{ is an equal }k\text{-clustering of }\bfX \},
\end{equation*}
and given a nonnegative integer $B$,
\begin{equation*}
\opt(\bfX,k,B)=\min\{\cost_p^B(X_1,\ldots,X_k)\mid \{X_1,\ldots,X_k\}\text{ is an equal }k\text{-clustering of }\bfX \}.
\end{equation*}

\medskip
We conclude this section by the observation that, given vectors $\bfc_1,\ldots,\bfc_k\in \mathbb{R}^d$, we can find an equal $k$-clustering $\{X_1,\ldots,X_k\}$ that minimizes  $\sum_{i=1}^kcost_p(X_i,\bfc_i)$  using a reduction to the classical \probPM problem on bipartite graphs that is well-known to be solvable in polynomial time.
Recall that a \emph{matching} $M$ of a graph $G$ is a set of edges without common vertices. It is said that a matching $M$ \emph{saturates} a vertex $v$ if $M$
has an edge incident to $v$. A matching $M$ is \emph{perfect} if every vertex of $G$ is saturated. The task of \probPM is, given a bipartite graph $G$ and a weight function $w\colon E(G)\rightarrow \mathbb{Z}_{\geq 0}$, find a perfect matching $M$ (if it exists) such that its weight $w(M)=\sum_{e\in M}w(e)$ is minimum. The proof of the following lemma essentially repeats the proof of Lemma~1 of~\cite{FominGPS21} but we provide it here for completeness.

\begin{lemma}\label{lem:means-clusters}
Let $\bfX=\{\bfx_1,\ldots,\bfx_n\}$ be a collection of points of $\mathbb{Z}^d$ and $k$ be a positive integer such that $n$ is divisible by $k$. Let  also $\bfc_1,\ldots,\bfc_k\in \mathbb{R}^d$. Then an equal 
$k$-clustering $\{X_1,\ldots,X_k\}$ of minimum  $\cost(X_1,\ldots,X_k,\bfc_1,\ldots,\bfc_k)$ can be found in polynomial time. 
\end{lemma}

\begin{proof}
Given $\bf X$ and $\bfc_1,\ldots,\bfc_k$, we construct the bipartite graph $G$ as follows. Let $s=\frac{n}{k}$ .
\begin{itemize}
\item For each $i\in\{1,\ldots,k\}$,  we construct a set of $s$ vertices $V_i=\{v_1^i,\ldots,v_s^i\}$ corresponding to the median $c_i$. Denote $V=\bigcup_{i=1}^kV_i$. 
\item For each $i\in\{1,\ldots,n\}$, construct a vertex $u_i$ corresponding to the vector $\bfx_i$ of $\bfX$ and make $u_i$ adjacent to the vertices of $V$. Denote $U=\{u_1,\ldots,u_n\}$.
\end{itemize}
We define the edge weights as follows.
\begin{itemize}
\item For every $i\in\{1,\ldots,n\}$ and $j\in \{1,\ldots,k\}$, set $w(u_iv_h^j)=||\bfc_j-\bfx_i||_p$ for $h\in\{1,\ldots,s\}$, that is, the weight of all edges joining $u_i$ corresponding to $\bfx_i$ with the vertices of $V_j$ corresponding to the median $\bfc_j$ are the same and coincide with the $\ell_p$ distance between $\bfx_i$ and $\bfc_j$.
\end{itemize}
Observe that $G(U,V)$ is a complete bipartite graph, where $U$ and $V$ form the bipartition. Note also that $|U|=|V|=n$.

Notice that we have the following one-to-one correspondence between  perfect matchings of $G$ and $k$-clusterings of $\bfX$. 
 In the forward direction, assume that $M$ is a perfect matching of $G$. We construct the clustering $\{X_1,\ldots,X_k\}$ as follows. 
 For every $h \in \{1,\ldots,n\}$, $u_h$ is saturated by $M$ and, therefore, there are $i_h\in\{1,\ldots,k\}$ and $j_h \in \{1,\ldots s\}$ such that edge $u_hv^{i_h}_{j_h}\in M$. 
We cluster the vectors of $\bfX$ according to $M$. Formally, we place $x_h$ in $X_{i_h}$ for each $h\in\{1,\ldots,n\}$. 
Clearly, $\{X_1,\ldots,X_k\}$ is a partition of $\{\bfx_1,\ldots,\bfx_n\}$ and $|X_i|=s$ for all $i\in\{1,\ldots,k\}$. By the definition of weights of the edges of $G$, 
$\cost_p(X_1,\ldots,X_k,c_1,\ldots,c_k)=w(M)$.  For the reverse direction, consider an equal $k$-clustering $\{X_1,\ldots,X_k\}$  of  $\bfX$.  Let $i \in \{1,\ldots,k\}$. 
Consider the cluster $X_i$ and assume that  $X_i=\{j_1,\ldots,j_s\}$.
Denote by  $M_i=\{u_{j_1}v_1^i,\ldots,u_{j_{s}}v_s^i\}$. Clearly, $M_i$ is a matching saturating the  vertices of $V_i$. 
We construct $M_i$ for every $i\in\{1,\ldots,k\}$ and set  $M=\bigcup_{i=1}^{k}M_i$. Since $\{X_1,\ldots,X_k\}$ is a partition of $\{1,\ldots,n\}$, $M$ is a matching saturating every vertex of $U$. 
By the definition of the weight of edges, $w(M)=\cost_p(X_1,\ldots,x_k,c_1,\ldots,c_k)$. Thus, finding a $k$-clustering  $\{X_1,\ldots,X_k\}$ that minimizes 
$\cost_p(X_1,\ldots,x_k,c_1,\ldots,c_k)$ is equivalent to computing a perfect matching of minimum weight in $G$.  
Then, because a perfect matching of minimum weight in $G$ can be found in polynomial time~\cite{Kuhn55,LovaszP09}, a $k$-clustering of minimum cost can be found in polynomial time.
This completes the proof of the lemma.
\end{proof}

\section{Lossy Kernel}\label{sec:kernelization} 

In this section, we prove Theorem~\ref{thm:lossykernel1} by establishing a  $2$-approximate polynomial kernel for  \probPEClust. 
 In Subsection~\ref{sec:tech}, we provide some auxiliary results, and in Subsection~\ref{subsec:2-appr},
we prove the main results. 
 Throughout this section we assume that $p\geq 0$ defining the  $\ell_p$-norm is a fixed constant.

\subsection{Technical lemmata}\label{sec:tech}
 We start by proving the following results about medians of clusters when their size is sufficiently big with respect to the budget.

\begin{lemma}\label{lem:big-clust1}
Let $\{X_1,\ldots,X_k\}$ be an equal $k$-clustering of a collection of points $\bfX=\{\bfx_1,\ldots,\bfx_n\}$ of $\mathbb{Z}^d$ 
of cost at most $B\in \mathbb{Z}_{\geq 0}$, and let $s=\frac{n}{k}$.
Then each cluster $X_i$ for $i\in\{1,\ldots,k\}$  contains at least $s-2B$ identical points. 
\end{lemma}

\begin{proof}
The claim is trivial if $s \leq 2B+1$.  Let $s \geq 2B+2$.
Assume to the contrary that a cluster $X_i$ has at most $s-2B-1$ identical points  for some $i \in \{1,\ldots,k\}$.
Let $\bfc_1,\ldots,\bfc_k$ be optimum medians for the clusters $X_1,\ldots,X_k$, respectively. Then we have that 
$\cost_p(X_1,\ldots,X_k)=\cost_p(X_1,\ldots,X_k,\bfc_1,\ldots,\bfc_k)$. 

Let $\bfx_{i_0} \in X_i$ be a point at minimum distance from $\bfc_i$.
Since there are at most $s-2B-1$ points in $X_i$ which are equal to $x_{i_0}$,
there are  $t=2B+1$ points $\bfx_{i_1},\ldots,\bfx_{i_t} \in X_i$ distinct from $x_{i_0}$.
Observe that 
\begin{equation}{\label{eqn:medianscoincideswithpoints1}}
\sum_{\bfx_h \in X_i}||\bfc_i-\bfx_h||_p \geq \sum_{j =0}^t||\bfc_i-\bfx_{i_j}||_p=||\bfc_i-\bfx_{i_0}||_p + \sum_{j=1}^{t}||\bfc_i-\bfx_{i_j}||_p.
\end{equation} 
We have two possibilities: either  $\bfc_i = \bfx_{i_0}$ or $\bfc_i \neq \bfx_{i_0}$.

If $\bfc_i = \bfx_{i_0}$, then $||\bfc_i-\bfx_{i_0}||_p=0$ and $||\bfc_i-\bfx_{i_j}||_p \geq 1$ for $j \in \{1,\ldots,t\}$, because $\bfx_{i_0},\bfx_{i_1},\ldots,\bfx_{j_t}$ have integer coordinates and $\bfx_{i_0}\neq \bfx_{i_j}$ for $j \in \{1,\ldots,t\}$. 
 Then from~(\ref{eqn:medianscoincideswithpoints1}), we get 
\begin{align*}
\sum_{\bfx_h \in X_i}||\bfc_i-\bfx_h||_p \geq \sum_{j =1}^t||\bfc_i-\bfx_{i_j}||_p \geq t = 2B+1>B,
\end{align*} 
 which is a contradiction with $\cost_p(X_1,\ldots,X_k)\leq B$.

If $\bfc_i \neq \bfx_{i_0}$, then $||\bfc_i-\bfx_{i_0}||_p > 0$. 
Because the points have integer coordinates and by the triangle inequality,
\begin{equation}{\label{eqn:medianscoincideswithpoints2}}
1 \leq ||\bfx_{i_0}-\bfx_{i_j}||_p \leq ||\bfx_{i_0}-\bfc_i||_p+||\bfx_{i_j}-\bfc_{i}||_p 
\end{equation}  
for every $j\in\{1,\ldots,t\}$.
Since  $\bfx_{i_0}$ is a point of $X_i$ at minimum distance from $\bfc_i$, 
\begin{equation}{\label{eqn:medianscoincideswithpoints3}}
||\bfx_{i_0}-\bfc_i||_p+||\bfx_{i_j}-\bfc_{i}||_p \leq 2||\bfx_{i_j}-\bfc_{i}||.
\end{equation}  
From~(\ref{eqn:medianscoincideswithpoints2}) and (\ref{eqn:medianscoincideswithpoints3}), we get 
$||\bfx_{i_j}-\bfc_{i}|| \geq \frac{1}{2}$ for $j\in\{1,\ldots,t\}$. 
Thus from~(\ref{eqn:medianscoincideswithpoints1}), we get 
\begin{align*}{\label{eqn2:medianscoincideswithpoints}}
\sum_{\bfx_h \in X_i}||\bfc_i-\bfx_h||_p \geq & \sum_{j =0}^t||\bfc_i-\bfx_{i_j}||_p=||\bfc_i-\bfx_{i_0}||_p + \sum_{j=1}^{t}||\bfc_i-\bfx_{i_j}||_p\\
 >&\sum_{j=1}^{t}||\bfc_i-\bfx_{i_j}||_p \geq \frac{1}{2}t =\frac{1}{2}(2B+1)> B,
\end{align*} 
which is a contradiction with $\cost_p(X_1,\ldots,X_k)\leq B$.
This completes the proof.
\end{proof}

\begin{lemma}\label{lem:big-clust2}
Let $\{X_1,\ldots,X_k\}$ be an equal $k$-clustering of a collection of points $\bfX=\{\bfx_1,\ldots,\bfx_n\}$ of $\mathbb{Z}^d$ 
of cost at most $B\in \mathbb{Z}_{\geq 0}$, and let $s=\frac{n}{k}\geq 4B+1$. Let also $\bfc_1,\ldots,\bfc_k \in \mathbb{R}^d$ be optimum medians for $X_1,\ldots,X_k$, respectively. 
Then for every $i \in \{1,\ldots,k\}$,  $\bfc_i=\bfx_{j}$ for a unique $\bfx_j \in \bfX_i$ such that $\bfX_i$ contains at least $s-2B$ points identical to $\bfx_j$. 
\end{lemma}

\begin{proof}
Consider a cluster $X_i$ with the median $\bfc_i$ for arbitrary $i \in \{1,\ldots,k\}$.
Since $s \geq 4B+1$, then by Lemma~\ref{lem:big-clust1}, there is $\bfx_j\in X_i$ such that $X_i$ contains at least 
$s-2B$ points identical to $\bfx_j$.  Notice that $\bfx_j$ is unique, because 
$X_i$ can contain at most $s-(s-2B)=2B$ distinct from $\bfx_j$ points, and since $s\geq 4B+1$, $s-2B\geq 2B+1>2B$.  We show that $\bfc_i=\bfx_j$. 

The proof is by contradiction. Assume that $\bfc_i \neq \bfx_j$. Let 
$S\subseteq\{1,\ldots,n\}$ be the set of indices of the points $\bfx_h\in X_i$ that coincide with $\bfx_j$, and denote by 
$T$ be the set of indices of the remaining points in $X_i$. 
We know that $|T| \leq 2B <|S|$ because $s \geq 4B+1$ and $|S| \geq 2B+1$.
Then 
\begin{equation}\label{eqn:equation1}
\begin{aligned} 
\cost_p(X_i)=&\cost_p(X_i,\bfc_i)=\sum_{h\in X_i}||\bfc_i-\bfx_h||_p=\sum_{h\in S}||\bfc_i-\bfx_h||_p+\sum_{h \in T}||\bfc_i-\bfx_{h}||_p\\ 
= &(|S|-|T|)||\bfc_i-\bfx_j||_p +\sum_{h \in T}(||\bfc_i-\bfx_j||+||\bfc_i-\bfx_{h}||_p). 
\end{aligned}
\end{equation}
On using the triangle inequality, we get 
\begin{equation}{\label{eqn:equation2}}
 (|S|-|T|)||\bfc_i-\bfx_j||_p +\sum_{h\in T}(||\bfc_i-\bfx_j||_p+||\bfc_i-\bfx_{h}||_p)
 \geq (|S|-|T|)||\bfc_i-\bfx_j||_p+\sum_{h\in T}||\bfx_j-\bfx_{h}||_p.
\end{equation}
We know that $(|S|-|T|)||\bfc_i-\bfx_j||_p>0$ because $|S|>|T|$ and $\bfc_i \neq\bfx_j$.
Then by~(\ref{eqn:equation2}), we have
\begin{equation}\label{eqn:equation3}
(|S|-|T|)||\bfc_i-\bfx_j||_p+\sum_{h \in T}||\bfx_j-\bfx_{h}||_p >\sum_{h\in T}||\bfx_j-\bfx_{h}||_p.
\end{equation}
Combining (\ref{eqn:equation1})--(\ref{eqn:equation3}), we conclude that $\cost_p(X_i)>\sum_{h\in T}||\bfx_j-\bfx_{h}||_p$. 
Let $\bfc'_i=\bfx_j$.  
Then
\begin{align*}
\cost_p(X_i,\bfc_i')=&\sum_{h\in X_i}||\bfc'_i-\bfx_h||_p=\sum_{h\in S}||\bfc'_i-\bfx_h||_p+\sum_{h \in T}||\bfc'_i-\bfx_{h}||_p=\sum_{h \in T}||\bfc'_i-\bfx_{h}||_p\\
<&~\cost_p(X_i)
\end{align*}
which contradicts that $\bfc_i$ is an optimum median for $X_i$. This concludes the proof.
\end{proof}

We use the following lemma to identify medians.

\begin{lemma}\label{lem:init-mean}
Let $\{X_1,\ldots,X_k\}$ be an equal $k$-clustering of a collection of points $\bfX=\{\bfx_1,\ldots,\bfx_n\}$ of $\mathbb{Z}^d$ 
of cost at most $B\in \mathbb{Z}_{\geq 0}$, and let $s=\frac{n}{k}\geq 4B+1$. suppose that $Y\subseteq \bfX$ is a collection of at least $B+1$ identical points of $\bfX$. 
Then there is $i\in\{1,\ldots,k\}$ such that an optimum median of $X_i$ coincides with $\bfx_j$ for $\bfx_j\in Y$.
\end{lemma}   
 
\begin{proof}
Let $\bfc_1,\ldots.\bfc_k$ be optimum medians of $X_1,\ldots,X_k$, respectively.  
Since $s\geq 4B+1$, then by Lemma~\ref{lem:big-clust2}, for every $i \in \{1,\ldots,k\}$,  $\bfc_i$  coincides with some element $\bfx_{h}$ of the cluster $X_i$.
For the sake of contradiction, assume that $\bfc_1,\ldots,\bfc_k$  are distinct from $\bfx_j \in Y$. This means that $\|\bfx_j-\bfc_i\|_p\geq 1$, because  the coordinates of the points of $\bfX$ are integer. 
Then
\begin{equation*}
\cost_p(X_1,\ldots,X_k)=\sum_{i=1}^k\cost_p(X_i,\bfc_i) \geq \sum_{i=1}^k\sum_{\bfx_h \in Y \cap X_i}||\bfc_i-\bfx_h||_p\geq  \sum_{i=1}^k|X_i\cap Y|=|Y|\geq B+1 >B,
\end{equation*}
 contradicting that $\cost_p(X_1,\ldots,X_k)\leq B$. This proves the lemma. 
\end{proof}

We use our next lemma to upper bound the clustering cost if we collect $s=\frac{n}{k}$ identical points in the same cluster.

\begin{lemma}\label{lem:exchange}   
Let $\{X_1,\ldots,X_k\}$ be an \emph{equal $k$-clustering} of a collection of points $\bfX=\{\bfx_1,\ldots,\bfx_n\}$ of $\mathbb{Z}^d$, and let $\bfc_1,\ldots,\bfc_k\in \mathbb{R}^d$. 
Suppose that $S$ is a collection of $s=\frac{n}{k}$ identical points of $\bfX$ and $\bfx_j\in S$. 
Then there is an equal $k$-clustering $\{X_1',\ldots,X_k'\}$ of $\bfX$ with $X_1'=S$ such that 
$$\cost_p(X_1',\ldots,X_k',\bfc_1',\ldots,\bfc_k')\leq \cost_p(X_1,\ldots,X_k,\bfc_1,\ldots,\bfc_k)+s||\bfc_1-\bfx_j||_p,$$
where $\bfc_1'=\bfx_j$ and $\bfc_h'=\bfc_h$ for $h\in\{2,\ldots,k\}$.
\end{lemma}
    
 \begin{proof}
 The claim is trivial if $S=X_1$, because we can set $X_i'=X_i$ for $i\in\{1,\ldots,k\}$. Assume that this is not the case and there are elements of $S$ that are not in $X_1$, and denote by $\bfx_{i_1},\ldots,\bfx_{i_t}$ these elements. We assume that $\bfx_{i_h} \in X_{i_h'}$, for $h \in \{1,\ldots,t\}$ for $i_h' \geq 2$. Because $|S|=s$, there are $\bfx_{j_1},\ldots,\bfx_{j_t} \in X_1$ such that $\bfx_{j_1},\ldots,\bfx_{j_t} \notin S$. We construct $X_1',\ldots,X_k'$ from $X_1,\ldots,X_k$ by exchanging the points $\bfx_{j_h}$ and $\bfx_{i_h}$ between $X_1$ and $X_{i_h'}$
for every $h \in \{1,\ldots,t\}$. Notice that $|X_1'|=\cdots=|X_k'|$, because the exchanges do not modify the sizes of the clusters. Thus,  $\{X_1',\ldots,X_k'\}$ is an equal $k$-clustering. 
We claim that $\{X_1',\ldots,X_k'\}$ satisfies the required property.

We have that 
\begin{equation}\label{eqn:diff}
\begin{aligned}
&\cost(X'_1,\ldots,X'_k,\bfc'_1,\ldots,\bfc'_k) - \cost(X_1,\ldots,X_k,\bfc_1,\ldots,\bfc_k)\\
&=\sum_{i=1}^{k}\sum_{\bfx_h \in X'_i}^{}||\bfx_h-\bfc'_i||_p-\sum_{i=1}^{k}\sum_{\bfx_h \in X_i}^{}||\bfx-\bfc_i||_p\\
&= \sum_{\bfx_h \in X'_1}||\bfx_h-\bfc'_1||_p-\sum_{\bfx_h \in X_1}||\bfx_h-\bfc_1||_p+\sum_{i=2}^{k}\big(\sum_{\bfx_h \in X'_i}||\bfx_h-\bfc'_i||_p -\sum_{\bfx_h \in X_i}||\bfx_h-\bfc_i||_p\big).
\end{aligned}
\end{equation}
Note that $\sum_{\bfx_h \in X'_1}||\bfx_h-\bfc'_1||_p=0$ and $\sum_{\bfx_h \in X_1}||\bfx_h-\bfc_1||_p \geq \sum_{h=1}^{t}||\bfx_{j_h}-\bfc_1||_p$.
Also by the construction of $X'_1,\ldots,X'_k$ and because $\bfc_{i}=\bfc'_{i}$ for $i \in \{2,\ldots,k\}$, we have that
 \begin{align*}
 \sum_{i=2}^{k}\big(\sum_{\bfx_h \in X'_i}||\bfx_h-\bfc'_i||_p -\sum_{\bfx_h \in X_i}||\bfx_h-\bfc_i||_p\big)=&
 \sum_{h=1}^{t}||\bfx_{j_h}-\bfc'_{i_h}||_p -\sum_{h=1}^t||\bfx_{i_h}-\bfc_{i_h}||_p\\
 =&\sum_{h=1}^{t}||\bfx_{j_h}-\bfc_{i_h}||_p -\sum_{h=1}^t||\bfx_{i_h}-\bfc_{i_h}||_p.
 \end{align*}
 Then extending (\ref{eqn:diff}) and applying the triangle inequality twice, we obtain that 
 \begin{align*}
&\cost(X'_1,\ldots,X'_k,\bfc'_1,\ldots,\bfc'_k) - \cost(X_1,\ldots,X_k,\bfc_1,\ldots,\bfc_k)\\ 
& \leq - \sum_{h=1}^{t}||\bfx_{j_h}-\bfc_1||_p+\sum_{h=1}^{t}||\bfx_{j_h}-\bfc_{i_h}||_p 
  -\sum_{h=1}^{t}||\bfx_{i_h}-\bfc_{i_h}||_p \\
  & = \sum_{h=1}^{t}\big(-||\bfx_{j_h}-\bfc_1||_p+||\bfx_{j_h}-\bfc_{i_h}||_p-||\bfx_{i_h}-\bfc_{i_h}||_p\big)
 \leq \sum_{h=1}^{t}\big(||\bfx_{i_h}-\bfc_{i_h}||_p-||\bfc_1-\bfc_{i_h}||_p\big) \\
 & \leq \sum_{h=1}^t||\bfx_{i_h}-\bfc_1||_p  \leq t||\bfx_j-\bfc_1||_p \leq s||\bfx_j-\bfc_1||_p
 \end{align*}
 as required by the lemma.
 \end{proof}

Our next lemma shows that we can solve   \probPEClust in polynomial time  if the cluster size is sufficiently big with respect to the budget.

\begin{lemma}\label{lem:polynomialtimealgorithm}
There is a polynomial time algorithm that, given a collection  $\bfX=\{\bfx_1,\ldots,\bfx_n\}$ of $n$ points of $\mathbb{Z}^d$, a positive integer $k$ such that $n$ is divisible by $k$, and a nonnegative integer $B$ such that $\frac{n}{k}\geq 4B+1$, either computes $\opt(X,k)\leq B$ and produces an equal $k$-clustering of minimum cost or correctly concludes that $\opt(\bfX,k)>B$.
\end{lemma}

\begin{proof}
Let $\bfX=\{\bfx_1,\ldots,\bfx_n\}$ be a collection of $n$ points of $\mathbb{Z}^d$ and let  $k$ be a positive integer such that $n$ is divisible by $k$, and suppose 
that  $s=\frac{n}{k}\geq 4B+1$ for a nonnegative integer $B$.

First, we exhaustively apply the following reduction rule. 

\begin{redrule}\label{red:ident}
If $\bfX$ contains a collection of $s$ identical points $S$, then set $\bfX:=\bfX\setminus S$ and $k:=k-1$.
\end{redrule}

To argue that the rule is safe, let $\bfX'=\bfX\setminus S$, where $S$ is a collection of $s$ identical points of $\bfX$, and let $k'=k$. Clearly, $\bfX'$ contains $n'=n-s$ points and $\frac{n'}{k'}=s$. 
If $\{X_1',\ldots,X_{k'}'\}$ is an equal $k'$-clustering of $\bfX'$, then $\{S,X_1',\ldots,X_{k'}'\}$ is an equal $k$-clustering of $\bfX$. Note that $\cost_p(S)=0$, because the elements of $S$ are identical.  Then
$\cost_p(S,X_1',\ldots,X_{k'}')=\cost_p(X_1',\ldots,X_{k'}')$. Therefore, $\opt(\bfX,k)\leq \opt(\bfX',k')$. We show that if $\opt(\bfX,k)\leq B$, then $\opt(\bfX,k)\geq \opt(\bfX',k')$. 

Suppose that $\{X_1,\ldots,X_k\}$ is an equal $k$-clustering of $\bfX$ with $\cost_p(X_1,\ldots,X_k)=\opt(X,k)\leq B$. Denote by $\bfc_1,\ldots,\bfc_k$ optimum medians of $X_1,\ldots,X_k$, respectively. Because $|S|=s\geq 4B+1\geq B+1$, there is a cluster whose optimum median is $\bfx_j$ for $\bfx_j\in S$. We assume without loss of generality that $X_1$ is such a cluster and  $\bfc_1=\bfx_j$. 
By Lemma~\ref{lem:exchange}, there is a $k$-clustering $\{S,X_2',\ldots,X_{k}'\}$ of $\bfX$ such that 
$\cost_p(S,X_2',\ldots,X_k',\bfc_1',\ldots,\bfc_k')\leq \cost_p(X_1,\ldots,X_k,\bfc_1,\ldots,\bfc_k)+s||\bfc_1-\bfx_j||_p$, 
where $\bfc_1'=\bfx_j$ and $\bfc_h'=\bfc_h$ for $h\in\{2,\ldots,k\}$. Because $\bfc_1=\bfx_j$, we conclude
that 
$\cost_p(X_2',\ldots,X_k')=\cost_p(S,X_2',\ldots,X_k',\bfc_1',\ldots,\bfc_k')\leq \cost_p(X_1,\ldots,X_k,\bfc_1,\ldots,\bfc_k)=\opt(X,k)$. 
Since $\{X_2',\ldots,X_k'\}$ is a $k'$-clustering of $\bfX'$, we have that $\opt(\bfX',k')\leq \cost_p(X_2',\ldots,X_k')\leq \opt(X,k)$ as required. 

We obtain that either $\opt(\bfX,k)=\opt(\bfX',k')\leq B$ or $\opt(\bfX,k)>B$ and $\opt(\bfX',k')> B$. 
Notice also that, given an optimum equal $k'$-clustering of $\bfX'$, we can construct the optimum $k$-clustering of $X$, by making $S$ a cluster.
Thus, it is sufficient to prove the lemma for the collection of points obtained by the exhaustive application of  
Reduction Rule~\ref{red:ident}. Note that if this collection is empty, then $\opt(\bfX,k)=0$ and the lemma holds. 
  This allows us to assume from now that $\bfX$ is nonempty and has no $s$ identical points. 

Suppose that $\{X_1,\ldots,X_k\}$ be an equal $k$-clustering with $\cost_p(X_1,\ldots,X_k)=\opt(\bfX,k)\leq B$. By Lemma~\ref{lem:big-clust2}, we have that  for every $i \in \{1,\ldots,k\}$, the optimum median $\bfc_i$ for $X_i$ is unique and $\bfc_i=\bfx_{j}$ for  $\bfx_j \in \bfX_i$ such that $\bfX_i$ contains at least $s-2B$ points identical to $\bfx_j$. Notice that $\bfc_1,\ldots,\bfc_k$ are pairwise distinct, because a collection of identical points cannot be split between distinct clusters in such a way that each of these cluster would contain at least $s-2B$ points. This holds because  
any collection of identical points of $\bfX$ contains at most $s-1$ elements and $2(s-2B)>s$ as $s\geq 4B+1$. By Lemma~\ref{lem:init-mean}, we have that if $\bfX$ contains a collection of identical points $S$ of size $B+1\leq s-2B$, then one of the optimum median should be equal to a point from $S$. 

These observations allow us to construct  (potential) medians $\bfc_1,\ldots,\bfc_t$ as follows: we iteratively compute inclusion maximal collections $S$ of identical points of $\bfX$ and if $|S|\geq B+1$, we set the next median $\bfc_i$ be equal to a point of $S$. If the number of constructed potential medians $t\neq k$, we conclude that $\bfX$ has no equal $k$-clustering of cost at most $B$. Otherwise, if $t=k$, we have that $\bfc_1,\ldots,\bfc_k$ should be optimum medians for an equal $k$-clustering of minimum cost if $\opt(\bfX,k)\leq B$. 

Then we compute in polynomial time an equal $k$-clustering $\{X_1,\ldots,X_k\}$ of $\bfX$ that minimizes $\cost_p(X_1,\ldots,X_k,\bfc_1,\ldots,\bfc_k)$  using Lemma~\ref{lem:means-clusters}.
If $\cost_p(X_1,\ldots,X_k,\bfc_1,\ldots,\bfc_k)>B$, then we conclude that $\opt(\bfX,k)>B$. Otherwise, we have that $\opt(\bfX,k)=\cost_p(X_1,\ldots,X_k,\bfc_1,\ldots,\bfc_k)$ and $\{X_1,\ldots,X_k\}$ is an equal $k$-clustering of minimum cost.
\end{proof}

Our next aim is to show that we can reduce the dimension and the absolute values of the coordinates of the points if $\opt(X,k)\leq B$. 
To achieve this, we mimic some ideas of the kernelization algorithm of Fomin et al. in~\cite{FominGP20} for the related clustering problem. 
However, they considered only points from $\{0,1\}^d$ and the Hamming norm. 

 \begin{lemma}{\label{lem:coordinatereduction1}}
 There is a polynomial time algorithm that, given a collection  $\bfX=\{\bfx_1,\ldots,\bfx_n\}$ of $n$ points of $\mathbb{Z}^d$, a positive integer $k$ such that $n$ is divisible by $k$, and a nonnegative integer $B$, either  correctly concludes that $\opt(\bfX,k)>B$ or computes a collection of $n$ points 
  $\bfY=\{\bfy_1,\ldots,\bfy_n\}$ of $\mathbb{Z}^{d'}$ such that the following holds:
\begin{enumerate}[label=(\roman*)]
\item For every partition $\{I_1,\ldots,I_k\}$ of $\{1,\ldots,n\}$ such that $|I_1|=\cdots=|I_k|=\frac{n}{k}$, either $\cost_p(X_1,\ldots,X_k)>B$ and  $\cost_p(Y_1,\ldots,Y_k)>B$ or 
  $\cost_p(X_1,\ldots,X_k)=\cost_p(Y_1,\ldots,Y_k)$, where $X_i=\{\bfx_h\mid h\in I_i\}$ and $Y_i=\{\bfy_h\mid h\in I_i\}$ for every $i\in\{1,\ldots,k\}$. 
 \item $d' = \mathcal{O}(kB^{p+1})$.
\item  $|\bfy_i[h]| = \mathcal{O}(kB^2)$ for $h \in \{1,\ldots,d'\}$ and $i \in\{1,\ldots,n\}$.
\end{enumerate}
\end{lemma}

\begin{proof}
Let   $\bfX=\{\bfx_1,\ldots,\bfx_n\}$ be a collection of $n$ points of $\mathbb{Z}^d$ and let $k$ be a positive integer such that $n$ is divisible by $k$. Let also $B$ be a nonnegative integer. 

We iteratively construct the partition $S=\{S_1,\ldots,S_t\}$ of $\{\bfx_1,\ldots,\bfx_n\}$ using the following greedy algorithm. 
Let $j \geq 1$ be an integer and suppose that the sets $S_0,\ldots,S_{j-1}$ are already constructed assuming that $S_0=\emptyset$.
Let $\bfZ=\{\bfx_1,\ldots,\bfx_n\} \setminus \cup_{i=0}^{j-1}S_i$. 
If $\bfZ = \emptyset$, then the construction of $S$ is completed.
If $\bfZ \neq \emptyset$, we construct $S_j$ as follows: 
\begin{itemize}
\item set $S_j:=\{\bfx_h\}$ for arbitrary $\bfx_h \in \bfZ$ and set $\bfZ:=\bfZ \setminus \{\bfx_h\}$,
 \item while there is $\bfx_r \in \bfZ$ such that $||\bfx_r-\bfx_{r'}||_p\leq B$ for some $\bfx_{r'} \in S_j$,
  set $S_j:=S_j\cup \{\bfx_r\}$ and set $\bfZ=\bfZ \setminus \{\bfx_r\}$.
\end{itemize}

The crucial property of the partition $S$ is that every cluster of an equal $k$-clustering of cost at most $B$ is entirely in some part of the partition.

\begin{claim}{\label{claim:clustercontainsinpartition}}
Let $\{X_1,\ldots,X_k\}$ be an equal $k$-clustering of $\bfX$ of cost at most $B$. Then for every $i \in \{1,\ldots,k\}$ there is $j \in \{1,\ldots,t\}$ such that $X_i \subseteq S_j$.
\end{claim}

\begin{proof}[Proof of Claim~\ref{claim:clustercontainsinpartition}]
Denote $\bfc_1,\ldots,\bfc_k \in \mathbb{R}^d$ the optimum medians for the clusters $X_1,\ldots,X_k$, respectively.
 Assume to the contrary that there is a cluster $X_i$ such that $\bfx_u,\bfx_v \in X_i$ with $\bfx_u$ and $\bfx_v$ in distinct collections of the partition $\{S_1,\ldots,S_t\}$. Then $||\bfx_u-\bfx_v||_p > B$ by the construction of $S_1,\ldots,S_t$ and 
\begin{equation*} 
 \cost_p(X_1,\ldots,X_k)\geq \cost_p(X_i)=\cost_p(X_i,\bfc_i)\geq \|\bfc_i-\bfx_u\|_p+\|\bfc_i-\bfx_v\|_p\geq \|\bfx_u-\bfx_v|\|_p > B
 \end{equation*} 
 contradicting that $\cost_p(X_1,\ldots,X_k)\leq B$.
\end{proof}

From the above Claim~\ref{claim:clustercontainsinpartition}, we have that if $t>k$, then $\bfX$ has no equal $k$-clustering of cost at most $B$, that is, $\opt(X,B)>B$. In this case, we return this answer and stop. From now on, we assume that this is not the case and $t \leq k$. 

By Lemma~\ref{lem:big-clust1}, at least $\frac{n}{k}-2B$ points in every cluster of an equal  $k$-clustering of cost at most $B$ are identical. Thus, if $\{X_1,\ldots,X_k\}$ is an equal $k$-clustering of cost at most $B$, then for each $i\in\{1,\ldots,k\}$, $X_i$ contains at most $2B+1$ distinct points. By Claim~\ref{claim:clustercontainsinpartition}, we obtain that for every $i\in\{1,\ldots,t\}$, $S_i$ should contain at most $k(2B+1)$ distinct points if $\bfX$ admits an equal $k$-clustering of cost at most $B$. Then for each $i\in\{1,\ldots,t\}$, we compute the number of distinct points in $S_i$ and if this number is bigger than  $k(2B+1)$, we conclude that  $\opt(\bfX,k)>B$. In this case we return this answer and stop. From now, we assume that this is not the case and each $S_i$ for $i\in\{1,\ldots,t\}$ contains at most $k(2B+1)$ distinct points. 

For a collection of points $Z\subseteq \bfX$, we say that a coordinate $h\in\{1,\ldots,d\}$ is \emph{uniform} for $Z$ if $\bfx_j[h]$ is the  same for all $\bfx_h\in Z$ and $h$ is  \emph{nonuniform} otherwise. 

Let $\ell_i$ be the number of nonuniform coordinates for $S_i$ for $i \in \{1,\ldots,t\}$, and let $\ell=\max_{1 \leq i \leq t}\ell_i$. For each $i\in\{1,\ldots,t\}$, we select a set of indices 
$R_i\subseteq \{1,\ldots,d\}$ of size $\ell$ such that $R_i$ contains all nonuniform coordinates for $S_i$. Note that $R_i$ may be empty if $\ell=0$. 
 We also define a set of coordinates $T_i=\{1,\ldots,d\}\setminus R_i$, for $i \in \{1,\ldots,t\}$.

For every $i \in \{1,\ldots,n\}$ and $j \in \{1,\ldots,t\}$ such that $  x_i \in S_j$, we define an  $(\ell+1)$-dimensional point $\bfx'_i$, where
$\bfx'_i[1,\ldots,\ell]=\bfx_i[R_j]$ and $\bfx'_i[\ell+1]=(j-1)(B+1)$. This way we obtain a collection of points 
 $\bf X'=\{\bfx'_1,\ldots,\bfx'_n\}$. For every  $j \in \{1,\ldots,t\}$, we define $S'_j=\{\bfx'_h \mid \bfx_h \in S_j\}$, that is, we construct the partition $S'=\{S'_1,\ldots,S'_t\}$ of $\{\bfx'_1,\ldots,\bfx'_n\}$ corresponding to $S$.

For each $i\in\{1,\ldots,t\}$, we do the following:
\begin{itemize}
    \item For each $h \in \{1,\ldots,\ell\}$, we find $M_h^{(i)}=\min\{\bfx'_j[h] \mid \bfx'_j \in S'_i\}$.
    \item For every $\bfx_j' \in S_i'$, we define a new point $\bfy_j$ by setting  $\bfy_j[h]=\bfx'_j[h]-M_h^{(i)}$ for $h \in \{1,\ldots,\ell\}$ and $\bfy_j[\ell+1]=\bfx'_j[\ell+1]=(j-1)(B+1)$.
\end{itemize}
This way, we construct the collection $\bfY=\{\bfy_1,\ldots,\bfy_n\}$ of points from $\mathbb{Z}^{\ell+1}$. Our algorithm return this collection of the points.

It is easy to see that the described algorithm runs in polynomial time. We show that  if the algorithm outputs $\bfY$, then this collection of the points satisfies conditions (i)--(iii) of the lemma. 

\medskip
To show (i), let $\{I_1,\ldots,I_k\}$ 
be a partition  of $\{1,\ldots,n\}$ such that $|I_1|=\cdots=|I_k|=\frac{n}{k}$, and let $X_i=\{\bfx_h\mid h\in I_i\}$ and $Y_i=\{\bfy_h\mid h\in I_i\}$ for every $i\in\{1,\ldots,k\}$. 
We show that 
either $\cost_p(X_1,\ldots,X_k)>B$ and  $\cost_p(Y_1,\ldots,Y_k)>B$ or $\cost_p(X_1,\ldots,X_k)=\cost_p(Y_1,\ldots,Y_k)$. 

Suppose that  $\cost_p(X_1,\ldots,X_k)\leq B$. Consider $i\in\{1,\ldots, k\}$ and denote by $\bfc_i$ the optimum median for $X_i$. By Claim~\ref{claim:clustercontainsinpartition}, there is $j\in\{1,\ldots,t\}$ such that $X_i\subseteq S_j$.  We define $\bfc_i'\in\mathbb{R}^{\ell+1}$ by setting $\bfc_i'[1,\ldots,\ell]=\bfc_i[R_j]$ and $\bfc_i'[\ell+1]=(j-1)(B+1)$. Further, we consider $\bfc_i''\in \mathbb{R}^{\ell+1}$ such that $\bfc_i''[h]=\bfc_i'[h]-M_h^{(j)}$ for $h\in\{1,\ldots,\ell\}$ and $\bfc_i''[\ell+1]=(j-1)(B+1)$.
Then by the definitions of $\bfX_i'$ and $\bfY_i$, we have that 
\begin{equation*}
\cost_p(X_i)=\cost_p(X_i,\bfc_i)=\cost_p(X_i',\bfc_i')=\cost_p(Y_i,\bfc_i'')\geq\cost_p(Y_i).
\end{equation*}
This implies that $\cost_p(X_1,\ldots,X_k)\geq \cost_p(Y_1,\ldots,Y_k)$. 

For the opposite direction, assume that $\cost_p(X_1,\ldots,X_k)\leq B$.
Similarly to $S'$, for every  $j \in \{1,\ldots,t\}$, we define $S''_j=\{\bfy_h \mid \bfx_h \in S_j\}$, that is, we construct the partition $S''=\{S''_1,\ldots,S''_t\}$ of $\bfY$ corresponding to $S$. 
We claim that for each $i\in\{1,\ldots,k\}$, there is $j\in\{1,\ldots,t\}$ such that $Y_i\subseteq S_j$.

The proof is by contradiction and is similar to the proof of Claim~\ref{claim:clustercontainsinpartition}. Assume that 
there is $i\in\{1,\ldots,k\}$ such that there are $\bfy_u,\bfy_v \in Y_i$ belonging to distinct sets of $S''$.   
Then $||\bfy_u-\bfy_v||_p\geq |\bfy_u[\ell+1]-\bfy_v[\ell+1]| > B$ by 
the construction of $S_1'',\ldots,S_t''$. Then 
\begin{equation*} 
 \cost_p(Y_1,\ldots,Y_k)\geq \cost_p(Y_i)=\cost_p(Y_i,\bfc_i)\geq \|\bfc_i-\bfx_u\|_p+\|\bfc_i-\bfx_v\|_p\geq \|\bfx_u-\bfx_v|\|_p > B,
 \end{equation*} 
 where $\bfc_i$ is an optimum median of $Y_i$. However, this contradicts that $\cost_p(Y_1,\ldots,Y_k)\leq B$.
 
 Consider $i\in\{1,\ldots, k\}$ and let $\bfc_i''\in\mathbb{R}^{\ell+1}$ an optimum median for $Y_i$. Let also $j\in\{1,\ldots,t\}$ be such that $Y_i\subseteq S_j$. 
 Notice that $\bfc_i''[\ell+1]=(j-1)(B+1)$ by the definition of $S_j$. We define 
 $\bfc_i'\in \mathbb{R}^{\ell+1}$ by setting $\bfc_i'[h]=\bfc_i''[h]+M_h^{(j)}$ for $h\in\{1,\ldots,\ell\}$ and $\bfc_i'[\ell+1]=\bfc_i''[\ell+1]=(j-1)(B+1)$. Then we define 
 $\bfc_i\in \mathbb{R}^d$, we setting $\bfc_i[R_j]=\bfc_i'[1,\ldots,\ell]$ and $\bfc_i[T_j]=\bfx_h[T_j]$ for arbitrary $\bfx_h\in S_j$. Because the coordinates in $T_j$ are uniform for $S_j$, the values in each coordinate $h\in T_j$ of the coordinates of the points of $X_i$ are the same. This implies that 
 \begin{equation*}
\cost_p(X_i)\leq\cost_p(X_i,\bfc_i)=\cost_p(X_i',\bfc_i')=\cost_p(Y_i,\bfc_i'')=\cost_p(Y_i).
\end{equation*} 
Hence, $\cost_p(X_1,\ldots,X_k)\leq \cost_p(Y_1,\ldots,Y_k)$.  
This completes the proof of (i).
 
 \medskip
 To show (ii), we prove that $\ell\leq kB^p(2B+1)$. For this, we show that $\ell_i\leq kB^p(2B+1)$ for every $i\in\{1,\ldots,t\}$. Consider $i\in\{1,\ldots,t\}$. Recall that $S_i$ contains at most $k(2B+1)$ distinct points.  Denote by $\bfx_{j_1},\ldots,\bfx_{j_r}$ the distinct points in $X_i$ and assume that they are numbered in the order in which they are included in $S_i$ by the greedy procedure constructing this set. 

Let $Z_q=\{\bfx_{j_1},\ldots,\bfx_{j_q}\}$ for $q\in\{1,\ldots,r\}$. We claim that $Z_q$ has at most $(q-1)B^p$ nonuniform coordinates for each $q\in\{1,\ldots,r\}$. The proof is by induction. The claim is trivial if $q=1$. Let $q>1$ and assume that the claim is fulfilled for $Z_{q-1}$. By the construction of $S_i$, $\bfx_{j_q}$ is at distance at most $B$ from $\bfx_{j_h}$ for some $h\in\{1,\ldots,q-1\}$. Then because 
$\|\bfx_{j_q}-\bfx_{j_h}\|_p\leq B$, we obtain that the points $\bfx_{i_q}$ and $\bfx_{i_h}$ differ in at most $B^p$ coordinates by the definition of the $\ell_p$-norm. Then because $Z_{q-1}$ has at most $(q-2)B^p$ nonuniform coordinates, $Z_q$ has at most $(q-1)B^p$ nonuniform coordinates as required.

Because the number of nonuniform coordinates for $S_i$ is the same as the number of nonuniform coordinates for $Z_r$ and $r\leq k(2B+1)$, we obtain that $\ell_i\leq kB^p(2B+1)$. Then
$\ell=\max_{1\leq i\leq t}\ell_i\leq  kB^p(2B+1)$. Because the points of $\bfY$ are in $\mathbb{Z}^{\ell+1}$, we have the required upper bound for the dimension. This concludes the proof of (ii).
 
Finally, to show (iii), we again exploit the property that every $S_i$ contains at most $k(2B+1)$ distinct points. Let $i\in\{1,\ldots,t\}$ and $h\in\{1,\ldots,d\}$ and denote by $\bfx_{j_1},\ldots,\bfx_{j_r}$ the distinct points in $X_i$.  Let $h\in\{1,\ldots,d\}$. We can assume without loss of generality that $\bfx_{j_1}[h]\leq\cdots\leq \bfx_{j_r}[h]$. 
We claim that $\bfx_{j_r}[h]-\bfx_{j_1}[h]\leq B(k(2B+1)-1)$. This is trivial if $r=1$. Assume that $r>1$. Observe that $\bfx_{j_q}[h]-\bfx_{j_{q-1}}[h]\leq B$ for $q\in\{2,\ldots,r\}$.
Otherwise, if there is $q\in\{2,\ldots,r\}$ such that $\bfx_{j_q}[h]-\bfx_{j_{q-1}}[h]> B$, then the distance from any point in $\{\bfx_{j_1},\ldots,\bfx_{j_{q-1}}\}$ to any point in $\{\bfx_{j_q},\ldots,\bfx_{j_{r}}\}$ is more than $B$ but this contradicts that these points are the distinct points of $S_i$. Then because $\bfx_{j_q}[h]-\bfx_{j_{q-1}}[h]\leq B$ for $q\in\{2,\ldots,r\}$ and $r\leq k(2B+1)$, we obtain that 
 $\bfx_{j_r}[h]-\bfx_{j_1}[h]\leq B(k(2B+1)-1)$. 
 
Then, by the definition of $\bfx_1',\ldots,\bfx_n'$, we obtain that for every $\bfx_{q}',\bfx_r'\in S_i'$ for some $i\in\{1,\ldots,t\}$ and every $h\in\{1,\ldots,\ell\}$, $|\bfx_q'[h]-\bfx_r'[h]|\leq B(k(2B+1)-1)$. By the definition of $M_h^{(i)}$ for $i\in\{1,\ldots,t\}$, we obtain that $|\bfy_j[h]|\leq B(k(2B+1)-1)$ for every $j\in\{1,\ldots,n\}$ and  every $h\in\{1,\ldots,\ell\}$. Because $|\bfy_j[\ell+1]|\leq (k-1)(B+1)$, we 
have that  $|\bfy_i[h]|\leq B(k(2B+1)-1)$ for $h \in \{1,\ldots,d'\}$ and $i \in\{1,\ldots,n\}$. This completes the proof of (iii) and the proof of the lemma.
\end{proof}

Finally in this subsection, we show the following lemma that is used to upper bound the additional cost incurred by the greedy clustering of blocks of identical points. 

\begin{lemma}\label{lem:bound-opt}
Let  $\bfX=\{\bfx_1,\ldots,\bfx_n\}$ be a collection of $n$ points of $\mathbb{Z}^d$ and set $k$ be  a positive integer such that $n$ is divisible by $k$. Suppose that $S_1,\ldots,S_t$ are disjoint collections of identical points of $\bfX$ such that $|S_1|=\cdots=|S_t|=\frac{n}{k}$ and $\bfY=\bfX\setminus \big(S_1\cup\cdots\cup S_t\big)$. Then $\opt(\bfY,k-t)\leq2\cdot \opt(\bfX,k)$.
\end{lemma}

\begin{proof}
Let $\{X_1,\ldots,X_k\}$ be an optimum equal $k$-clustering of $\bfX$ with optimum medians $\bfc_1,\ldots,\bfc_k$ of $X_1,\ldots,X_k$, respectively, that is,
$\opt(\bfX,k)=\cost_p(X_1,\ldots,X_k)=\cost_p(X_1,\ldots,X_k,\bfc_1,\ldots,\bfc_k)$.
Let $\bfx_{i_h}\in\bfS_h$ for $h\in \{1,\ldots,t\}$. Consider a $t$-tuple of $(j_1,\ldots,j_t)$ of distinct indices from $\{1,\ldots,k\}$ such that 
\begin{equation}\label{eqn:choice}
\|\bfx_{i_1}-\bfc_{j_1}\|_p+\cdots+\|\bfx_{i_t}-\bfc_{j_t}\|_p=\min_{(q_1,\ldots,q_t)}\big(\|\bfx_{i_1}-\bfc_{q_1}\|_p+\cdots+\|\bfx_{i_t}-\bfc_{q_t}\|_p\big), 
\end{equation} 
where the minimum in the right part is taken over all $t$-tuples $(q_1,\ldots,q_t)$ of distinct indices from $\{1,\ldots,k\}$. 
Denote $\ell=k-t$. Iteratively applying Lemma~\ref{lem:exchange} for $S_1,\ldots,S_t$ and the medians $\bfc_{j_1},\ldots,\bfc_{j_t}$, we obtain that there is an equal $\ell$-clustering $\{Y_1,\ldots,Y_\ell\}$ of $\bfY$ such that 
\begin{equation}\label{eqn:iter}
\cost_p(S_1,\ldots,S_t,Y_1,\ldots,Y_\ell)\leq\cost_p(X_1,\ldots,X_k)+s\sum_{h=1}^t \|\bfx_{i_h}-\bfc_{j_h}\|_p.
\end{equation}
Because the points in each $S_i$ are identical, $\cost_p(S_i)=0$ and, therefore, $\cost_p(S_1,\ldots,S_t,Y_1,\ldots,Y_\ell)=\cost_p(Y_1,\ldots,Y_k)$. Then by (\ref{eqn:iter}),
\begin{equation}\label{eqn:upper-bound-one}
\opt(\bfY,\ell)\leq \cost_p(Y_1,\ldots,Y_k)\leq \opt(\bfX,k)+s\sum_{h=1}^t \|\bfx_{i_h}-\bfc_{j_h}\|_p.
\end{equation}
This implies that to prove the lemma, it is sufficient to show that 
\begin{equation}\label{eqn:upper-bound-two}
s\sum_{h=1}^t \|\bfx_{i_h}-\bfc_{j_h}\|_p\leq \opt(\bfX,k). 
\end{equation}

To prove (\ref{eqn:upper-bound-two}), we consider the following auxiliary  clustering problem. Let $\bfZ=S_1\cup\cdots\cup S_t$ and $s=\frac{n}{k}$. The task of the problem is to find a partition $\{Z_1,\ldots,Z_k\}$ of $\bfZ$, where some sets may be empty and $|Z_i|\leq s$ for every  $i\in\{1,\ldots,k\}$, such that 
\begin{equation}\label{eqn:opt-aux}
\sum_{i=1}^t\sum_{\bfx_h\in Z_i}\|\bfx_h-\bfc_i\|_p
\end{equation}
is minimum. In words, we cluster the elements of $\bfZ$ in optimum way into clusters of size at most $s$ using the optimum medians $\bfc_1,\ldots,\bfc_k$ for the clustering $\{X_1,\ldots,X_k\}$. 
Denote by $\opt^*(\bfZ,k)$ the minimum value of (\ref{eqn:opt-aux}). Because in this problem the task is to cluster a subcollection of points of $\bfX$ and we relax the cluster size constraints, we have that 
$\opt^*(\bfZ,k)\leq \opt(\bfX,k)$.  We show the following claim.

\begin{claim}\label{cl:aux-opt}
\begin{equation*}
\opt^*(\bfZ,k)\geq s\cdot\min_{(q_1,\ldots,q_t)}\sum_{h=1}^t\|\bfx_{i_h}-\bfc_{q_h}\|_p,
\end{equation*}
where the minimum is taken over all $t$-tuples $(q_1,\ldots,q_t)$ of distinct indices from $\{1,\ldots,k\}$.
\end{claim}

\begin{proof}[Proof of Claim~\ref{cl:aux-opt}]
We show that the considered auxiliary clustering problem can be reduced to the \textsc{Min Cost Flow} problem (see, e.g., the textbook of Kleinberg and Tardos~\cite{KleinbergT06} for the introduction)\footnote{Equivalently one may use the ILP statement.}. We construct the directed graph $G$ and define the \emph{cost} and \emph{capacity} functions $c(\cdot)$ and $\omega(\cdot)$ on the set of arcs $A(G)$ as follows.
\begin{itemize}
\item Construct two vertices $a$ and $b$ that are the \emph{source} and \emph{target} vertices, respectively.
\item For every $i\in\{1,\ldots,t\}$, construct a vertex $u_i$ (corresponding to $S_i$) and an arc $(a,u_i)$ with $\omega(a,u_i)=0$.
\item For every $j\in \{1,\ldots,k\}$, construct a vertex $v_j$ (corresponding to $Z_j$) and and arc $(v_j,b)$ with $\omega(v_j,b)=0$.
\item For every $h\in\{1,\ldots,t\}$ and every $j\in\{1,\ldots,k\}$, construct an arc $(u_h,v_j)$ and set $\omega(u_i,v_j)=\|\bfx_{i_h}-\bfc_j\|_p$ (recall that $\bfx_{i_h}\in S_h$).
\item For every arc $e$ of $G$, set $c(e)=s$, where $s=\frac{n}{k}$. 
\end{itemize}
Then the volume of a \emph{flow} $f\colon A(G)\rightarrow \mathbb{R}_{\geq 0}$ is $v(f)=\sum_{i=1}^tf(a,u_i)$ and its \emph{cost} is $\omega(f)=\sum_{a\in A(G)}\omega(a)\cdot f(a)$. 
Let $f^*(\cdot)$ be a flow of volume $st$ with minimum cost. We claim that $\omega(f^*)=\opt^*(\bfZ,k)$. 

Assume that $\{Z_1,\ldots,Z_k\}$ is a partition of $\bfZ$ such that $|Z_i|\leq s$ for every  $i\in\{1,\ldots,k\}$ and 
$\opt^*(\bfZ,k)=\sum_{i=1}^t\sum_{\bfx_h\in Z_i}\|\bfx_h-\bfc_i\|_p$. We define the flow $f(\cdot)$ as follows:
\begin{itemize}
\item for every $i\in\{1,\ldots,t\}$, set $f(a,u_i)=s$,
\item for every $i\in\{1,\ldots,t\}$ and $j\in\{1,\ldots,k\}$, set $f(u_i,v_j)=|S_i\cap Z_j|$, and
\item for every $j\in\{1,\ldots,t\}$, set $f(v_j,b)=|Z_j|$.
\end{itemize}
It is easy to verify that $f$ is a feasible flow of volume $st$ and $\omega(f)=\sum_{i=1}^t\sum_{\bfx_h\in Z_i}\|\bfx_h-\bfc_i\|_p$. Thus, $\omega(f^*)\leq \omega(f)=\opt^*(\bfZ,k)$.

For the opposite inequality, consider $f^*(\cdot)$. By the well-known property of flows (see~\cite{KleinbergT06}), we can assume that $f^*(\cdot)$ is an integer flow, that is, $f^*(e)$ is a nonnegative integer for every $e\in A(G)$. Since $v(f^*)=st$, we have that $f^*(a,u_i)=s$ for every $i\in\{1,\ldots,t\}$.  Then we construct the clustering  $\{Z_1,\ldots,Z_k\}$ as follows: for every $i\in\{1,\ldots,i\}$ and $j\in\{1,\ldots,k\}$, we put exactly $f^*(u_i,v_j)$ points of $S_i$ into $Z_j$. Because  $f^*(a,u_i)=s$ for every $i\in\{1,\ldots,t\}$ and $c(v_j,b)=s$ for every $j\in\{1,\ldots,k\}$, we obtain that $\{Z_1,\ldots,Z_k\}$ is a partition of $\bfZ$ such that $|Z_i|\leq s$ for every  $i\in\{1,\ldots,k\}$ and $\sum_{i=1}^t\sum_{\bfx_h\in Z_i}\|\bfx_h-\bfc_i\|_p=\omega(f^*)$. This implies that 
$\opt^*(\bfZ,k)\leq \sum_{i=1}^t\sum_{\bfx_h\in Z_i}\|\bfx_h-\bfc_i\|_p=\omega(f^*)$. 

This proves that $\omega(f^*)=\opt^*(\bfZ,k)$. Moreover, we can observe that, given an integer flow $f(\cdot)$ with $v(f)=st$, we can construct a feasible clustering $\{Z_1,\ldots,Z_k\}$ of cost $\omega(f)$ such that for every $i\in\{1,\ldots,t\}$ and every $j\in\{1,\ldots,k\}$,   $|S_i\cap Z_j|=f(u_i,v_j)$. Recall that the capacities of the arcs of $G$ are the same and are equal to $s$. Then again exploiting the properties of flows (see~\cite{KleinbergT06}), we observe that there is a flow $f^*(\cdot)$ with $v(f^*)=st$ of minimum cost such that \emph{saturated} arcs (that is, arcs $e$ with $f^*(e)=c(e)=s$) compose internally vertex disjoint $(a,b)$-paths, and the flow on other arcs is zero. This implies, that for the clustering $\{Z_1,\ldots,Z_k\}$ constructed for $f^*(\cdot)$, for every $j\in\{1,\ldots,k\}$, ether $Z_j=\emptyset$ or there is $i\in\{1,\ldots,t\}$ such that $Z_j=S_i$. 
Assume that $j_1,\ldots,j_t$ are distinct indices from $\{1,\ldots,k\}$ such that $Z_{j_h}=S_h$ for $h\in\{1,\ldots,t\}$. Then  
$\omega(f^*)=\sum_{i=1}^t\sum_{\bfx_h\in Z_i}\|\bfx_h-\bfc_i\|_p=s\sum_{h=1}^t\|\bfx_{i_h}-\bfc_{j_h}\|_p$ and
\begin{equation*} 
\opt^*(\bfZ,k)=\omega(f^*)=s\sum_{h=1}^t\|\bfx_{i_h}-\bfc_{j_h}\|_p\geq  s\cdot\min_{(q_1,\ldots,q_t)}\sum_{h=1}^t\|\bfx_{i_h}-\bfc_{q_h}\|_p,
\end{equation*}
where the minimum is taken over all $t$-tuples $(q_1,\ldots,q_t)$ of distinct indices from $\{1,\ldots,k\}$. This proves the claim.
\end{proof}
Recall that $\opt^*(\bfZ,k)\leq \opt(\bfX,k)$.
By the choice of $j_1,\ldots,j_t$ in (\ref{eqn:choice}) and Claim~\ref{cl:aux-opt}, we obtain that inequality (\ref{eqn:upper-bound-two}) holds. Then by (\ref{eqn:upper-bound-two}), we have that  $\opt(\bfY,k-t)\leq 2\cdot\opt(\bfX,k)$ as required by the lemma. 
\end{proof}

\subsection{Proof of Theorem~\ref{thm:lossykernel1}}\label{subsec:2-appr}
Now we are ready to show the result about approximate kernel that we restate.

\Lossy*

\begin{proof}
Let $(\bfX,k,B)$ be an instance of \probPEClust with \linebreak $\bfX=\{\bfx_1,\ldots,\bfx_n\}$, where the points are from  $\mathbb{Z}^d$  and $n$ is divisible by $k$.
 Recall that a lossy kernel consists of two algorithms. The first algorithm   is 
  a polynomial time reduction producing an instance $(\bfX',k',B')$ of bounded size. The second algorithm is a solution-lifting  and  for every equal $k'$-clustering $\{X_1',\ldots,X_k'\}$ of $\bfX'$, this algorithm produces in polynomial time an equal $k$-clustering $\{X_1,\ldots,X_k\}$ of $\bfX$ such that 
\begin{equation}\label{eqn:condition}
\frac{\cost_p^B(X_1,\ldots,X_k)}{\opt(\bfX,k,B)}\leq 2\cdot \frac{\cost_p^{B'}(X_1',\ldots,X_{k'}')}{\opt(\bfX',k',B')}. \footnote{Note that by our simplifying assumption, $\frac{\cost_p^B(X_1,\ldots,X_k)}{\opt(\bfX,k,B)}=1$ if $\opt(\bfX,k,B)=0$, and the same assumption is used for $ \frac{\cost_p^{B'}(X_1',\ldots,X_{k'}')}{\opt(\bfX',k',B')}$.}
\end{equation}

We separately consider the cases when $\frac{n}{k} \geq 4B+1$ and $\frac{n}{k} \leq 4B$.

Suppose that $\frac{n}{k} \geq 4B+1$. Then we apply the algorithm from Lemma~\ref{lem:polynomialtimealgorithm}. If the algorithm returns the answer that $\bfX$ does no admit an equal $k$-clustering of cost at most $B$, then the reduction algorithm returns an trivial no-instance $(\bfX',k',B')$ of constant size, that is, an instance  such that $\bfX'$ has no clustering of cost at most $B'$. For example, we set $\bfX'=\{(0),(1)\}$, $k'=1$, and $B'=0$. 
Here and in the further cases when the reduction algorithm returns a trivial no-instance, the solution-lifting algorithm returns an arbitrary equal $k$-clustering of $\bfX$. 
Since $\cost_p^B(X_1,\ldots,X_k)=\opt(\bfX,k,B)=B+1$, 
 (\ref{eqn:condition}) holds.
Assume that the algorithm from Lemma~\ref{lem:polynomialtimealgorithm} produced an equal $k$-clustering $\{X_1,\ldots,X_k\}$ of minimum cost. Then the  reduction
 returns an arbitrary instance of \probPEClust of constant size. 
  For example, we can use $\bfX'=\{(0)\}$, $k'=1$, and $B'=0$.
 The solution-lifting algorithms  always returns $\{X_1,\ldots,X_k\}$. Clearly, $\cost_p^B(X_1,\ldots,X_k)=\opt(\bfX,k,B)$ and (\ref{eqn:condition}) is fulfilled.

From now on, we assume that $\frac{n}{k} \leq 4B$, that is, $n \leq 4Bk$. 
 We apply the algorithm from Lemma~\ref{lem:coordinatereduction1}. If this algorithm reports that there is no  equal $k$-clustering of cost at most $B$, then the reduction algorithm returns a trivial no-instance and the solution-lifting algorithm outputs an arbitrary equal $k$-clustering of $\bfX$. Clearly, (\ref{eqn:condition}) is satisfied. Assume that this is not the case. Then we obtain a 
  collection of $n\leq 4Bk$ points  $\bfY=\{\bfy_1,\ldots,\bfy_n\}$ of $\mathbb{Z}^{d'}$ satisfying conditions (i)--(iii) of Lemma~\ref{lem:coordinatereduction1}. That is,
  \begin{itemize}
 \item[(i)] for every partition $\{I_1,\ldots,I_k\}$ of $\{1,\ldots,n\}$ such that $|I_1|=\cdots=|I_k|=\frac{n}{k}$, either $\cost_p(X_1,\ldots,X_k)>B$ and  $\cost_p(Y_1,\ldots,Y_k)>B$ or 
  $\cost_p(X_1,\ldots,X_k)=\cost_p(Y_1,\ldots,Y_k)$, where $X_i=\{\bfx_h\mid h\in I_i\}$ and $Y_i=\{\bfy_h\mid h\in I_i\}$ for every $i\in\{1,\ldots,k\}$, 
 \item[(ii)] $d' = \mathcal{O}(kB^{p+1})$, and 
  \item[(iii)] $|\bfy_i[h]| = \mathcal{O}(kB^2)$ for $h \in \{1,\ldots,d'\}$ and $i \in\{1,\ldots,n\}$. 
\end{itemize}
By (i), for given an equal $k$-clustering clustering $\{Y_1,\ldots,Y_k\}$ of $\bfY$, we can compute the corresponding clustering $\{X_1,\ldots,X_k\}$ by setting $X_i=\{\bfx_h\mid \bfy_h\in Y_i\}$ for $i\in\{1,\ldots,k\}$. Then $\opt(\bfX,k,B)=\opt(\bfY,k,B)$
and 
\begin{equation}\label{eqn:x-y}
\frac{\cost_p^B(X_1,\ldots,X_k)}{\opt(\bfX,k,B)}= \frac{\cost_p^{B}(Y_1,\ldots,Y_{k})}{\opt(\bfY,k,B)}. 
\end{equation}
 Hence the instances $(\bfX,k,B)$ and $(\bfY,k,B)$ are equivalent. We continue with the compressed instance $(\bfY,k,B)$.
 
Now we apply the greedy procedure that constructs clusters $S_1,\ldots,S_t$ composed by identical points. Formally, we initially set $\bfX':=Y$, $k':=k$, and $i:=0$. Then we do  the following:
\begin{itemize}
\item while $\bfX'$ contains a collections $S$ of $s$ identical points, set $i:=i+1$, $S_i:=S$, $\bfX':=\bfX'\setminus S$, and $k':=k'-1$.
\end{itemize}

 Denote by $\bfX'$ the set of points obtained by the  application of the procedure and let $S_1,\ldots,S_t$ be the collections of identical points constructed by the procedure.
 Note that $k'=k-t$. We also define $B'=2B$. Notice that it may happen that $\bfX'=\bfY$ or $\bfX'=\emptyset$.
 The crucial property exploited by the kernelization is that  by  Lemma~\ref{lem:bound-opt}, $\opt(\bfX',k')\leq 2\cdot \opt(\bfY,k)$. 
 
We argue that if $k'>B$, then we have no $k$-clustering of cost at most $B$. 
 Suppose that $k'>B'$. Consider an arbitrary equal $k'$-clustering $\{X_1',\ldots,X_{k'}'\}$ of $\bfX'$. Because the construction of $S_1,\ldots,S_t$ stops when there is no collection of $s$ identical points, 
 each cluster $X_i'$ contains at least two distinct points. Since all points have integer coordinates, we have that $\cost_p(X_i')\geq 1$ for every $i\in\{1,\ldots,k'\}$. Therefore, 
 $\cost_p(X_1',\ldots,X_{k'}')=\sum_{i=1}^{k'}\cost_p(X_i')\geq k'>B'=2B$.  This means that $2\cdot \opt(Y,k)\geq \opt(\bfX',k')>2B$ and $\opt(\bfY,k)>B$. Using this, our reduction 
 algorithm returns a trivial no-instance. Then the solution-lifting algorithm outputs an arbitrary equal $k$-clustering of $\bfX$ and this satisfies (\ref{eqn:condition}).  
 
From now on we assume that $k'\leq B'=2B$ and construct the reduction and solution lifting algorithms for this case. 

If $k'=0$, then $\bfX'=\emptyset$ and the reduction algorithm simply returns an arbitrary instance of constant size. Otherwise, 
 our reduction algorithms returns $(\bfX',k',B')$. 
Observe that since  $k'\leq B'=2B$,  $|\bfX'|\leq n\leq 4B^2$. Recall that  $d'=\Oh(B^{p+2})$ and
$|\bfx_i'[h]| = \mathcal{O}(B^3)$ for $h \in \{1,\ldots,d'\}$ for every point $\bfx_i'\in \bfX'$. 
We conclude that the instance $(\bfX',k',B')$ of  \probPEClust satisfies the size conditions of the theorem.

Now we describe the solution-lifting algorithm and argue that inequality (\ref{eqn:condition}) holds.  

If $k'=0$, then the solution-lifting algorithm ignores the output of the reduction algorithm which was arbitrary. It takes the equal $k$-clustering $\{S_1,\ldots,S_k\}$ of $\bfY$ and outputs the equal $k$-clustering $\{X_1,\ldots,X_k\}$ of $\bfX$  by 
setting $X_i=\{\bfx_h\mid \bfy_h\in S_i\}$ for $i\in\{1,\ldots,k\}$.
Clearly, $\cost_p(S_1,\ldots,S_k)=\cost_p(X_1,\ldots,X_p)=0$. Therefore,  (\ref{eqn:condition}) holds.

If 
$k'>0$, we  
consider an equal $k'$-clustering $\{X_1',\ldots,X_{k'}'\}$ of $\bfX'$. The solution-lifting algorithm constructs an equal  $k$-clustering $\{S_1,\ldots,S_t,X_1',\ldots,X_{k'}'\}$, that is, we just add the clusters constructed by our greedy procedure. 
Since the points in each set $S_i$ are identical, $\cost_p(S_i)=0$ for every $i\in\{1,\ldots,t\}$. Therefore, 
\[\cost_p(S_1,\ldots,S_t,X_1',\ldots,X_{k'}')=\cost_p(X_1',\ldots,X_{k'}').\] Notice that since $\opt(\bfX',k')\leq 2\cdot \opt(\bfY,k)$, we have that 
 $\opt(\bfX',k',B')\leq 2\cdot \opt(\bfY,k,B)$. Indeed, if $\opt(\bfY,k)\leq B$, then $\opt(\bfX',k')\leq 2B=B'$. Hence, 
 $\opt(\bfY,k,B)=\opt(\bfY,k)$, $\opt(\bfX',k',B')=\opt(\bfX',k')$, and  $\opt(\bfX',k',B')\leq 2\cdot \opt(\bfY,k,B)$.
 If   $\opt(\bfY,k)>B$, then $\opt(\bfY,k,B)=B+1$. In this case 
 $2\cdot \opt(\bfY,k,B)=2B+2>\opt(\bfX',k',B')$, because $\opt(\bfX',k',B')\leq B'+1=2B+1$. Finally,  since 
$\cost_p(S_1,\ldots,S_t,X_1',\ldots,X_{k'}')=\cost_p(X_1',\ldots,X_{k'}')$ and $\opt(\bfX',k',B')\leq 2\cdot \opt(\bfY,k,B)$, we conclude that  
\begin{equation}\label{eqn:y-x}
\frac{\cost_p^B(S_1,\ldots,S_t,X_1',\ldots,X_{k'}')}{\opt(\bfY,k,B)}\leq 2\cdot\frac{\cost_p^{B}(X_1,\ldots,X_{k'}')}{\opt(\bfX',k',B')}. 
\end{equation}

Then the solution-lifting algorithm computes the equal $k$-clustering $\{X_1,\ldots,X_k\}$ for the equal $k$-clustering $\{Y_1,\ldots,Y_k\}=\{S_1,\ldots,S_t,X_1',\ldots,X_{k'}'\}$ of $\bfY$ by 
setting $X_i=\{\bfx_h\mid \bfy_h\in Y_i\}$ for $i\in\{1,\ldots,k\}$. Combining (\ref{eqn:x-y}) and (\ref{eqn:y-x}), we obtain (\ref{eqn:condition}). 

This concludes the description of the reduction and solution-lifting algorithms, as well as the proof of  their correctness. To argue that the reduction algorithm is a  polynomial time algorithm, we observe that the algorithms from Lemmata~\ref{lem:polynomialtimealgorithm} and \ref{lem:coordinatereduction1} run in polynomial time. Trivially, the greedy construction of $S_1,\ldots,S_t$, $\bfX$, and $k'$ can be done in polynomial time. Therefore, the reduction algorithm runs in polynomial time. The solution-lifting algorithm is also easily implementable to run in polynomial time. This concludes the proof. 
\end{proof}

\section{Kernelization}\label{sec:kern}
In this section we study (exact) kernelization of clustering with equal sizes. In Subsection~\ref{sec:lower-kern} we prove 
Theorem~\ref{thm:no-kern} claiming that  decision version of the problem,  \probEClust, does not admit a polynomial kernel being parameterized by $B$ only.
We also show  in Subsection~\ref{subsec:polyk} that the technical lemmata developed in the previous section for approximate kernel, can be used to   prove that 
 \probEClust parameterized by $k$ and $B$ admits a polynomial kernel.

\subsection{Kernelization lower bound}\label{sec:lower-kern}
In this section, we show  that it is unlikely that \probEClust admits a polynomial kernel when parameterized by $B$ only. We prove this for $\ell_0$ and $\ell_1$-norms.  Our lower bound holds even for points with binary coordinates, that is, for points from $\{0,1\}^d$. For this, we use the result of  Dell and Marx~\cite{DellM18} about kernelization lower bounds for the \probDM problem.

A hypergraph $\mathcal{H}$ is said to be \emph{$r$-uniform} for a positive integer $r$, if every hyperedge of $\mathcal{H}$ has size $r$. Similarly to graphs, a set of hyperedges $M$ is a \emph{matching} if the hyperedges in $M$ are pairwise disjoint, and $M$ is \emph{perfect} if every vertex of $\mathcal{H}$ is \emph{saturated} in $M$, that is, included in one of the hyperedges of $M$.  \probDM asks, given a $r$-uniform hypergraph $\mathcal{H}$, whether $\mathcal{H}$ has a perfect matching. Dell and Marx~\cite{DellM18} proved the following kernelization lower bound.

\begin{proposition}[\cite{DellM18}]\label{prop:dmatch}  
 Let $r\geq 3$ be an integer and let $\varepsilon$ be a positive real. If  $\nonopoly$, then \probDM 
 does not have kernels with $\Oh(\big(\frac{|V(\mathcal{H})|}{r}\big)^{r-\varepsilon})$ hyperedges.
 \end{proposition}   
 
 We need a weaker claim.
 
 \begin{corollary}\label{cor:dmatch}
 \probDM admits no polynomial kernel when parameterized by the number of vertices of the input hypergraph, unless  $\nonopoly$. 
 \end{corollary}  
   
\begin{proof} 
To see the claim, it is sufficient to observe that the existence of a polynomial kernel for \probDM parameterized by $|V(\mathcal{H})|$ implies that the problem has a kernel such that the number of hyperedges is polynomial in $|V(\mathcal{H})|$ with the degree of the polynomial that does not depend on $d$ contradicting Proposition~\ref{prop:dmatch}. 
\end{proof}  
 
We show the kernelization lower bound for $\ell_0$ and $\ell_1$ using the fact that optimum medians can be  computed by the \emph{majority rule} for a collection of binary points. Let $X$ be a collection of points of $\{0,1\}^d$. We construct $\bfc\in \{0,1\}^d$ as follows: for $i\in\{1,\ldots,d\}$, consider the multiset $S_i=\{\bfx[i]\mid \bfx\in X\}$ and set $\bfc[i]=0$ if at least half of the elements of $S_i$ are zeros, and set $\bfc[i]=1$ otherwise. It is straightforward to observe the following.

\begin{observation}\label{obs:majority}
Let $X$ be a collection of points of $\{0,1\}^d$ and let  $\bfc\in \{0,1\}^d$ be a vector constructed by the majority rule. Then for $\ell_0$ and $\ell_1$-norms, $\bfc$ is an optimum median for $X$.
\end{observation}

We also use the following lemma that is a special case of Lemma~9 of~\cite{FominGPS21}.

   \begin{lemma}[\cite{FominGPS21}]\label{lem:exchange-old}
 Let $\{X_1,\ldots,X_k\}$ be an equal $k$-clustering of a collection of points $\bfX=\{\bfx_1,\ldots,\bfx_n\}$ from $\{0,1\}^d$, and let  $\bfc_1,\ldots,\bfc_k$ be optimum medians for $X_1,\ldots,X_k$, respectively.  
 Let also $\bfC\subseteq\{\bfc_1,\ldots,\bfc_k\}$ be the set of medians coinciding  with some points of $\bfX$.
 Suppose that every collection of the same points of $\bfX$ has size at most $\frac{n}{k}$.
 Then there is an equal $k$-clustering $\{X_1',\ldots,X_k'\}$ of $\bfX$ such that 
 $\cost_0(X_1',\ldots,X_k',\bfc_1,\ldots,\bfc_k)\leq\cost_0(X_1,\ldots,X_k,\bfc_1,\ldots,\bfc_k)$ and for every 
 $i\in\{1,\ldots,k\}$, the following is fulfilled: if $\bfc_i\in \bfC$, then each $\bfx_h\in \bfX$ coinciding with $\bfc_i$ is in $X_i'$.
 \end{lemma}

   Now we are ready to prove Theorem~\ref{thm:no-kern}, we restate it here. 
   
 \Nokernel*
\begin{proof}  Notice that for any binary vector $\bfx\in\{0,1\}^d$, $\|\bfx\|_0=\|\bfx\|_1$. Since we consider only instances where the input points are binary, we can assume that the medians of clusters are binary as well by Observation~\ref{obs:majority}. Then it is sufficient to prove the theorem for one norm, say $\ell_0$.  
We reduce from \probDM. Let $\mathcal{H}$ be an $r$-uniform hypergraph.  
Denote by $v_1,\ldots,v_n$ the vertices and by $E_1,\ldots,E_m$ the hyperedges of $\mathcal{H}$, respectively. We assume that $n$ is divisible by $r$, as otherwise $\mathcal{H}$ has no perfect matching. We also assume that $r\geq 3$, because for $r\leq 2$, \probDM can be solved in polynomial time~\cite{LovaszP09}.

We construct the instance $(\bfX,k,B)$ of \probEClust, where $\bfX$ is a collection of $(r-1)n+rm$ points of $\{0,1\}^d$, where $d=2rn$.  

To describe the construction of $\bfX$, we partition the set  $\{1,\ldots,2rn\}$ of coordinate indices  into $n$ blocks $R_1,\ldots,R_n$ of size $2r$ each. For every $i\in\{1,\ldots,n\}$, we select an index $p_i\in R_i$ and set $R_i'=R_i\setminus\{p_1\}$. Formally,
\begin{itemize} 
\item $R_i=\{2r(i-1)+1,\ldots,2ri\}$ for $i\in\{1,\ldots,n\}$,
\item $p_i=2r(i-1)+1$ for $i\in\{1,\ldots,n\}$, and
\item $R_i'=\{2r(i-1)+2,\ldots,2ri\}$ for $i\in\{1,\ldots,n\}$.
\end{itemize}

The set of points $\bfX$ consists of $n+m$ blocks of identical points $V_1,\ldots,V_n$ and $F_1,\ldots,F_m$, where $|V_i|=r-1$ for each $i\in\{1,\ldots,n\}$ and $|F_i|=r$ for $i\in\{1,\ldots,m\}$. Each block $V_i$ is used to encode the vertex $v_i$, and each block $F_i$ is used to encode the corresponding hyperedge $E_i$. An example is shown in Figure~\ref{fig:no-kern}.

 For each $i\in\{1,\ldots,n\}$, we define the vector $\bfv_i\in\{0,1\}^{2rn}$ corresponding to the vertex $v_i$ of $\mathcal{H}$:
 \begin{equation*}   
\bfv_i[j]=
\begin{cases}
1&\mbox{if }j\in R_i,\\
0&\mbox{otherwise}.
\end{cases}
\end{equation*}
Then $V_i$ consists of $r-1$ copies of $\bfv_i$ that we denote $\bfv_i^{(1)},\ldots,\bfv_i^{(r-1)}$.

For every $j\in\{1,\ldots,m\}$, we define the vector $\bff_j\in \{0,1\}^{2rn}$ corresponding to the hyperedge $E_j=\{v_{i_1^{(j)}},\ldots,v_{i_r^{(j)}}\}$:
\begin{equation*}   
\bff_j[h]=
\begin{cases}
1&\mbox{if }h=p_s\text{ for some }s\in\{i_1^{(j)},\ldots,i_r^{(j)}\},\\
0&\mbox{otherwise}. 
\end{cases}   
\end{equation*}
Then $F_j$ includes $r$ copies of $\bff_j$ denoted by $\bff_j^{(1)},\ldots,\bff_j^{(r)}$.

\begin{figure}[!h]
\[
\bfX=
\left(
\begin{array}{cc|cc|cc|cc|cc|cc||ccc|ccc|ccc|ccc}
1&1&0&0&0&0&0&0&0&0&0&0&1&1&1&0&0&0&1&1&1&0&0&0\\
\hdashline
1&1&0&0&0&0&0&0&0&0&0&0&0&0&0&0&0&0&0&0&0&0&0&0\\
1&1&0&0&0&0&0&0&0&0&0&0&0&0&0&0&0&0&0&0&0&0&0&0\\
1&1&0&0&0&0&0&0&0&0&0&0&0&0&0&0&0&0&0&0&0&0&0&0\\
1&1&0&0&0&0&0&0&0&0&0&0&0&0&0&0&0&0&0&0&0&0&0&0\\
1&1&0&0&0&0&0&0&0&0&0&0&0&0&0&0&0&0&0&0&0&0&0&0\\
\hline
0&0&1&1&0&0&0&0&0&0&0&0&1&1&1&0&0&0&0&0&0&1&1&1\\
\hdashline
0&0&1&1&0&0&0&0&0&0&0&0&0&0&0&0&0&0&0&0&0&0&0&0\\
0&0&1&1&0&0&0&0&0&0&0&0&0&0&0&0&0&0&0&0&0&0&0&0\\
0&0&1&1&0&0&0&0&0&0&0&0&0&0&0&0&0&0&0&0&0&0&0&0\\
0&0&1&1&0&0&0&0&0&0&0&0&0&0&0&0&0&0&0&0&0&0&0&0\\
0&0&1&1&0&0&0&0&0&0&0&0&0&0&0&0&0&0&0&0&0&0&0&0\\
\hline
0&0&0&0&1&1&0&0&0&0&0&0&1&1&1&0&0&0&1&1&1&0&0&0\\
\hdashline
0&0&0&0&1&1&0&0&0&0&0&0&0&0&0&0&0&0&0&0&0&0&0&0\\
0&0&0&0&1&1&0&0&0&0&0&0&0&0&0&0&0&0&0&0&0&0&0&0\\
0&0&0&0&1&1&0&0&0&0&0&0&0&0&0&0&0&0&0&0&0&0&0&0\\
0&0&0&0&1&1&0&0&0&0&0&0&0&0&0&0&0&0&0&0&0&0&0&0\\
0&0&0&0&1&1&0&0&0&0&0&0&0&0&0&0&0&0&0&0&0&0&0&0\\
\hline
0&0&0&0&0&0&1&1&0&0&0&0&0&0&0&1&1&1&0&0&0&1&1&1\\
\hdashline
0&0&0&0&0&0&1&1&0&0&0&0&0&0&0&0&0&0&0&0&0&0&0&0\\
0&0&0&0&0&0&1&1&0&0&0&0&0&0&0&0&0&0&0&0&0&0&0&0\\
0&0&0&0&0&0&1&1&0&0&0&0&0&0&0&0&0&0&0&0&0&0&0&0\\
0&0&0&0&0&0&1&1&0&0&0&0&0&0&0&0&0&0&0&0&0&0&0&0\\
0&0&0&0&0&0&1&1&0&0&0&0&0&0&0&0&0&0&0&0&0&0&0&0\\
\hline
0&0&0&0&0&0&0&0&1&1&0&0&0&0&0&1&1&1&1&1&1&0&0&0\\
\hdashline
0&0&0&0&0&0&0&0&1&1&0&0&0&0&0&0&0&0&0&0&0&0&0&0\\
0&0&0&0&0&0&0&0&1&1&0&0&0&0&0&0&0&0&0&0&0&0&0&0\\
0&0&0&0&0&0&0&0&1&1&0&0&0&0&0&0&0&0&0&0&0&0&0&0\\
0&0&0&0&0&0&0&0&1&1&0&0&0&0&0&0&0&0&0&0&0&0&0&0\\
0&0&0&0&0&0&0&0&1&1&0&0&0&0&0&0&0&0&0&0&0&0&0&0\\
\hline
0&0&0&0&0&0&0&0&0&0&1&1&0&0&0&1&1&1&0&0&0&1&1&1\\
\hdashline
0&0&0&0&0&0&0&0&0&0&1&1&0&0&0&0&0&0&0&0&0&0&0&0\\
0&0&0&0&0&0&0&0&0&0&1&1&0&0&0&0&0&0&0&0&0&0&0&0\\
0&0&0&0&0&0&0&0&0&0&1&1&0&0&0&0&0&0&0&0&0&0&0&0\\
0&0&0&0&0&0&0&0&0&0&1&1&0&0&0&0&0&0&0&0&0&0&0&0\\
0&0&0&0&0&0&0&0&0&0&1&1&0&0&0&0&0&0&0&0&0&0&0&0\\
\end{array}
\right)
\]
\caption{The construction of $\bfX$ for $\mathcal{H}$ with $V(\mathcal{H})=\{v_1,\ldots,v_6\}$ and 
the hyperedges $E_1=\{v_1,v_2,v_3\}$, $E_2=\{v_4,v_5,v_6\}$, $E_3=\{v_1,v_3,v_5\}$, and $E_4=\{v_2,v_4,v_5\}$. The collection of the points $\bfX$ is shown here as a matrix, where each column is a point of $\bfX$.
Note that $r=3$ here. The blocks of $\bfX$ are shown by solid lines and the part of $\bfX$ corresponding to the vertices of $\mathcal{H}$ is separated from the part corresponding to hyperedges by a double line. The blocks of coordinates with indices $R_1,\ldots,R_6$  are separated by solid lines. 
The coordinates with the indices $p_1=1$, $p_2=7$, $p_3=13$, $p_4=19$, $p_5=25$, and $p_6=31$ are underlined by dashed lines.  
 }\label{fig:no-kern}
\end{figure}

 To complete the construction of the instance of \probEClust, we define
 \begin{itemize}
 \item $k=n+m-\frac{n}{r}$,
 \item $B=(3r-2)n$.
 \end{itemize}  
Recall that $n$ is divisible by $r$ and note that $\frac{(r-1)n+rm}{k}=r$.

It is straightforward to verify that the construction of $(\bfX,k,B)$ is polynomial.
   
We claim that the hypergraph $\mathcal{H}$ has a perfect matching if and only if $(\bfX,k,B)$ is a yes-instance of \probEClust. The proof uses the following property of the points of $\bfX$:
for every $i\in\{1,\ldots,n\}$ and every $j\in\{1,\ldots,m\}$,
\begin{equation}\label{eq:constr-cost}
\|\bfv_i-\bff_j\|_0=
\begin{cases}
3r-2&\mbox{if }v_i\in E_j,\\
3r&\mbox{if }v_i\notin E_j.
\end{cases}
\end{equation} 

 For the forward direction, assume that $\mathcal{H}$ has a perfect matching $M$. Assume without loss of generality that 
 $M=\{E_1,\ldots,E_s\}$ for $s=\frac{n}{r}$. 
 Since $M$ is a prefect matching, for every $i\in\{1,\ldots,n\}$, there is a unique $h_i\in\{1,\ldots,s\}$ such that $v_i\in E_{h_i}$.   
 We construct the equal $k$-clustering $\{X_1,\ldots,X_k\}$ as follows. 
 
 For every $i\in\{1,\ldots,n\}$, we define 
 $X_i=V_i\cup\{\bff_{h_i}^{(t)}\}$, where $t$ is chosen from the set $\{1,\ldots,r\}$ in such a way that $X_1,\ldots,X_n$ are disjoint. In words, we initiate each cluster $X_i$ by setting $X_i:=V_i$ for $i\in\{1,\ldots,n\}$. This way, we obtain $n$ clusters of size $r-1$ each.
 Then we consider the blocks of points $F_1,\ldots,F_s$ corresponding to the hyperedges  of $M$ and split them between the clusters $X_1,\ldots,X_n$ by including a single element  into each cluster. It is crucial that each $X_i=V_i$ is complemented by an element of $F_{h_i}$, that is, by an element of the initial cluster corresponding to the hyperedge saturating the vertex $v_i$. Since $M$ is a perfect matching, this splitting is feasible.  
 
 Notice that the first $s$ blocks of points $F_1,\ldots,F_s$ are split between $X_1,\ldots,X_n$. The remaining $m-s$ blocks $F_{s+1},\ldots, F_{m}$ have size $r$ each and form clusters $X_{n+1},\ldots,X_{k}$. This completes the construction of $\{X_1,\ldots,X_k\}$.
 
To evaluate $\cost_0(X_1,\ldots,X_k)$,  notice that the optimal median $\bfc_i=\bfv_i$ for $i\in\{1,\ldots,n\}$ by the majority rule. Then, by (\ref{eq:constr-cost}), 
$\cost_0(X_i)=\|\bfv_i-\bff_{h_i}\|=3r-2$.
Since the clusters $X_{n+1},\ldots,X_{r}$ consist of identical points, we have that $\cost_0(X_i)=0$ for $i\in\{1,\ldots,m-s\}$. 
Then $\cost(X_1,\ldots,X_k)=(3r-2)n=B$. Therefore, $(\bfX,k,B)$ is a yes-instance of \probEClust.

For the opposite direction, let $\{X_1,\ldots,X_k\}$ be an equal $k$-clustering of $\bfX$ of cost at most $B$. Denote by $\bfc_1,\ldots,\bfc_r$ the optimal medians constructed by the majority rule. Observe that the choice of a median by the majority rule described above is not symmetric, because if $i$-th coordinates of the points in a cluster have the same number of zeros and ones, the rule selects the zero value for the $i$-coordinate of the median.  We show the following claim.

\begin{claim}\label{cl:means}
For every $i\in\{1,\ldots,k\}$, either $\bfc_i\in\{\bfv_1,\ldots,\bfv_n\}$ or $\bfc_i[j]=0$ for all $j\in R_1'\cup\ldots\cup R_n'$. Moreover, the medians of the first type, that is, coinciding with one of $\bfv_1,\ldots,\bfv_n$, are distinct. 
\end{claim}
 
\begin{proof}[Proof of Claim~\ref{cl:means}]
Suppose that $\bfc_i[h]\neq 0$ for some  $h\in R_j'$, where $j\in\{1,\ldots,n\}$. Observe that, by the construction of $\bfX$, for every point $\bfx\in \bfX$,  
$\bfx[h]=1$ only if $\bfx\in V_j$. 
Since $\bfc_i$ is constructed by the majority rule, we obtain that more than half elements of $X_i$ are from $V_j$ and $\bfc_i=\bfv_j$. To see the second part of the claim, notice that $|V_j|=r-1$ and, therefore, at most one cluster $X_i$ of size $r$ can have at least half of its elements from $V_j$.  
\end{proof}
   
By Claim~\ref{cl:means}, we assume without loss of generality that $\bfc_i=\bfv_i$ for $i\in\{1,\ldots,\ell\}$ for some $\ell\in\{0,\ldots,r\}$ ($\ell=0$ if there is no cluster with the median from
$\{\bfv_1,\ldots,\bfv_n\}$) and $\bfc_i[j]=0$ for $j\in R_1'\cup\ldots\cup R_n'$ whenever $i\in\{\ell+1,\ldots,k\}$. 
Because the medians $\bfc_1,\ldots,\bfc_\ell$ are equal to points of $\bfX$, by Lemma~\ref{lem:exchange-old}, we can assume that $V_i\subset X_i$ for $i\in\{1,\ldots,\ell\}$.

\begin{claim}\label{cl:ell}
 $\ell=n$.
\end{claim} 
 
\begin{proof}[Proof of Claim~\ref{cl:ell}]
The proof is by contradiction. Assume that $\ell<n$. Consider the elements of $n-\ell$ blocks $V_{\ell+1},\ldots,V_n$. Let $p$ be the number of elements of $V_{\ell+1}\cup\ldots\cup V_n$ included in $X_1,\ldots,X_\ell$ and the remaining $q=(r-1)(n-\ell)-p$ elements are in $X_{\ell+1},\ldots,X_k$. 
By the definition of $\bfv_1,\ldots,\bfv_n$, 
if a point $\bfv_h^{(t)}\in V_h$ for some $h\in\{\ell+1,\ldots,n\}$  is in $X_i$ for some $i\in\{1,\ldots,\ell\}$, then 
$\|\bfv_h^{(t)}-\bfc_i\|_0=\|\bfv_h-\bfv_i\|_0=4r$.
Also we have that if $\bfv_h^{(t)}\in V_h$ for some $h\in\{\ell+1,\ldots,n\}$  is in $X_i$ for some $i\in\{\ell+1,\ldots,r\}$, then $\|\bfv_h^{(t)}-\bfc_I\|_0=\|\bfv_h-\bfc_I\|_0\geq |R_h'|=2r-1$.  
By (\ref{eq:constr-cost}), if the unique point $X_i\setminus V_i$ is $\bff_h^{(t)}\in F_h$ for some $h\in\{1,\ldots,m\}$, then 
$\|\bff_h^{(t)}- \bfc_i\|_0=\|\bff_h-\bfv_i\|_0\geq 3r-2$. 
Then $\sum_{i=1}^\ell\cost_0(X_i)\geq 4rp+(3r-2)(\ell-p)$ and $\sum_{i=\ell+1}^k\cost_0(X_i)\geq (2r-1)q$.
Recall also that $r\geq 3$ and, therefore, 
$r+2\leq 2r-1$ and $(r+2)(r-1)>3r-2$.
Summarizing, we obtain that 
\begin{align*}
\cost_0(X_1,\ldots,X_k)=&\sum_{i=1}^\ell\cost_0(X_i)+\sum_{i=\ell+1}^k\cost_0(X_i)\\
&\geq \big(4rp+(3r-2)(\ell-p)\big)+\big((2r-1)q\big)=(3r-2)\ell+(r+2)p+(2r-1)q\\
&\geq (3r-2)\ell+(r+2)(p+q)=(3r-2)\ell+(r+2)(r-1)(n-\ell)\\
&>(3r-2)n=B,
\end{align*} 
but this contradicts that $\cost_0(X_1,\ldots,X_k)\leq B$. This proves the claim.
\end{proof}   

By Claim~\ref{cl:ell}, we obtain that $\bfc_i=\bfv_i$ and $X_i\subset I_i$ for $i\in\{1,\ldots,n\}$.   For every $i\in \{1,\ldots,n\}$, $X_i\setminus V_i$ contains a unique point. Clearly, this is a point from 
$F_1\cup\cdots\cup F_m$. Denote by  $\bff_{h_i}^{(t_i)}$ the point of  $X_i\subset I_i$ for $i\in\{1,\ldots,n\}$. 
By (\ref{eq:constr-cost}), $\|\bfc_i-\bff_{h_i}^{(t_i)}\|_0=\|\bfc_i-\bff_{h_i}\|_0\geq 3r-2$ for every $i\in \{1,\ldots,n\}$. This means that 
\begin{equation*}
B\geq \cost_0(X_1,\ldots,X_k)=\sum_{i=1}^n\cost_0(X_i)+\sum_{i=n+1}^k\cost_0(X_i)
\geq \sum_{i=1}^n\cost_0(X_i)\geq (3d-2)n=B.
\end{equation*}
   Therefore, $\sum_{i=n+1}^k\cost_0(X_i)=0$. Hence, $k-n=m-s$ clusters $X_{n+1},\ldots,X_k\subseteq F_1\cup\cdots\cup F_m$, where $s=\frac{n}{r}$, consists of identical points.
Without loss of generality, we assume that   $F_{s+1},\ldots,F_m$ form these  clusters. Then the elements of $F_1,\ldots,F_s$ are split to complement $V_1,\ldots,V_n$ to form $X_1,\ldots,X_n$. 
In particular, for every $i\in\{1,\ldots,n\}$, there is $\bff_{h_i}^{(t_i)}\in X_i$ for some $h_i\in \{1,\ldots,m\}$ and $t_i\in\{1,\ldots,r\}$.   

We claim that $M=\{E_1,\ldots,E_s\}$  is a perfect matching of $\mathcal{H}$. To show this, consider a vertex $v_i\in V(\mathcal{H})$. We prove that $v_i\in E_{h_i}$. For sake of contradiction, assume that $v_i\notin E_{h_i}$. Then
$\|\bff_{h_i}^{(t_i)}-\bfc_i\|_0=\|\bff_{h_i}-\bfv_i\|_0=3r$ by (\ref{eq:constr-cost}) and
 \begin{equation*}
\cost_0(X_1,\ldots,X_k)=\sum_{j=1}^n\cost_0(X_j)\geq \sum_{j=1}^n\|\bff_{h_j}^{(t_j)}-\bfc_i\|_0= \sum_{j=1}^n\|\bff_{h_i}-\bfv_i\|_0\geq (3r-2)n+2>B;
\end{equation*}   
a contradiction with $\cost_0(X_1,\ldots,X_k)\leq B$.   Hence, every vertex of $V(\mathcal{H})$ is saturated by some hyperedge of $M$. Since $|M|=s=\frac{n}{r}$, we have that the hyperedges of $M$ are pairwise disjoint, that is, $M$ is a matching.   
Since every vertex is saturated and $M$ is a matching, $M$ is a perfect matching. 

This concludes the proof of our claim that  $\mathcal{H}$ has a perfect matching if and only if $(\bfX,k,B)$ is a yes-instance of \probEClust.   

Observe that $B=(3r-2)n$ in the reduction meaning that $B=\Oh(n^2)$. Since \probEClust is in \classNP, there is a polynomial reduction form \probEClust to \probDM.
Then if \probEClust has a polynomial kernel when parameterized by $B$, then \probDM has a polynomial kernel when parameterized by the number of vertices of the input hypergraph. This leads to a contradiction with Corollary~\ref{cor:dmatch} and completes the proof of the theorem.
\end{proof}   
 
 \subsection{Polynomial kernel for $k+B$ parameterization}\label{subsec:polyk}
In this subsection we prove Theorem~\ref{thm:kern}
  that we restate here.

\Kernel*
\begin{proof}
Let $(\bfX,k,B)$ be an instance of \probEClust with $\bfX=\{\bfx_1,\ldots,\bfx_n\}$, where the points are from  $\mathbb{Z}^d$. Recall that $n$ is divisible by $k$.

Suppose $\frac{n}{k}\geq 4B+1$. Then we can apply the algorithm from Lemma~\ref{lem:polynomialtimealgorithm}. If the algorithm returns that there is no equal $k$-clustering of cost at most $B$, then the kernelization algorithm returns a trivial no-instance of \probEClust. Otherwise, if $\opt(X,k)\leq B$, then the algorithm returns a trivial yes-instance. 

Assume from now that  $\frac{n}{k} \leq 4B$, that is, $n \leq 4Bk$. Then we apply the algorithm from Lemma~\ref{lem:coordinatereduction1}. If this algorithm reports that there is no  equal $k$-clustering of cost at most $B$, then the kernelization  algorithm returns a trivial no-instance of \probEClust. Otherwise, the algorithm from Lemma~\ref{lem:coordinatereduction1} returns a 
  collection of $n\leq 4Bk$ points  $\bfY=\{\bfy_1,\ldots,\bfy_n\}$ of $\mathbb{Z}^{d'}$ satisfying conditions (i)--(iii) of the lemma. By (i), we obtain that the instances $(\bfX,k,B)$ and $(\bfY,k,B)$ of  \probEClust are equivalent. By (ii), we have that the dimension $d'= \mathcal{O}(k(B^{p+1}))$, and by (iii), each coordinate of a point takes an absolute value of $\mathcal{O}(kB^2)$. Thus, $(\bfY,k,B)$ is a required kernel.
\end{proof}

\section{\textsf{APX}-hardness of \probMEClust}
\label{sec:apxhard}

In this section, we prove \textsf{APX}-hardness of \probMEClust w.r.t. Hamming ($\ell_0$) and $\ell_1$ distances. The constructed hard instances consists of high-dimensional binary ($0/1$) points. As the $\ell_0$ and $\ell_1$ distances between any two binary points are the same, we focus on the case of $\ell_0$ distances. Our reduction is from \textsc{3-Dimensional Matching} (\textsc{3DM}), where we are given three disjoint sets of elements $X, Y$ and $Z$ such that $|X|=|Y|=|Z|=n$ and a set of $m$ triples $T\subseteq X\times Y\times Z$. In addition, each element of $W:= X\cup Y\cup Z$ appears in at most 3 triples. A set $M \subseteq T$ is called a matching if no element of $W$ is contained in more than one triples of $M$. The goal is to find a maximum cardinality matching. We need the following theorem due to Petrank \cite{Patrank94}. 

\begin{theorem}[Restatement of Theorem 4.4 from 
\cite{Patrank94}] There exists a constant $0 <\gamma <
1$, such that it is NP-hard to distinguish the instances of the \textsc{3DM} problem
in which a perfect matching exists, from the instances in which there is a matching of size at most $(1-\gamma)n$.
\end{theorem}

$\gamma$ should be seen as a very small constant close to 0. 
We use the construction described in Section \ref{sec:lower-kern}, with a small modification. 

We are given an instance of \textsc{3DM}. Let $N=3n$, the total number of elements. We construct a binary matrix $A$ of dimension $6N\times (2N+3m)$. For each element we take $2$ columns and for each triple 3 columns. The $6N$ row indexes are partitioned into $N$ parts each of size 6. In particular, let $R_1=\{1,\ldots,6\}$, $R_2=\{7,\ldots,12\}$ and so on. For the $i$-th element, we construct the column $a_i$ of length $6N$ which has 1 corresponding to the indexes in $R_i$ and 0 elsewhere. 

Recall that each element can appear in at most 3 triples. For each element $x$, consider any arbitrary ranking of the triples that contain it. The occurrence of $x$ in a triple with rank $j$ is called its $j$-th occurence for $1\le j \le 3$. For example, suppose $x$ appears in three triples $t_w, t_y$ and $t_z$. One can consider the ranking $1. t_w, 2. t_y, 3. t_z$. Then, the occurence of $x$ in $t_y$ is called 2-nd occurence. Let $v_i^j$ be the $j$-th index of $R_i$ for $1\le i\le N, 1\le j\le 3$. For each triple $t$ with $j_1$-, $j_2$- and $j_3$-th occurences of the elements $p,q$ and $r$ in $t$, respectively, we construct the column $b_t$ of length $6N$ which has 1 corresponding to the indexes $v_p^{j_1}, v_q^{j_2}$ and $v_r^{j_3}$, and 0 elsewhere. 

The triple columns are defined in a different way in our reduction in Section \ref{sec:lower-kern} where for each triple and each element, a fixed index is set to 1. But, we set different indexes based on the occurences of the element. This ensures that for two different triple columns $b_s$ and $b_t$, their Hamming distance $d_H(b_s,b_t)=6$. Note that $d_H(a_i,b_t)=7$ if the element $i$ is in triple $t$, otherwise $d_H(a_i,b_t)=9$. Set cluster size to be 3, number of clusters $k$ to be $(2N/3)+m$. We will prove the following lemma. 

\begin{lemma}
If there is a perfect matching, there is a feasible clustering of cost $7N$. If all matchings have size at most $(1-\gamma)n$, any feasible clustering has cost at least $7(1-\gamma)N+(23/3)\gamma N$. 
\end{lemma}

Note that it is sufficient to prove the above lemma for showing the \textsf{APX}-hardness of the problem. The proof of the first part of the lemma is exactly same as in the previous construction. We will prove the second part. To give some intuition of the cost suppose there is a matching of the maximum size $(1-\gamma)n$. Then we can cluster the matched elements and triples in the same way as in the perfect matching case by paying a cost of $7(1-\gamma)N$. Now for each unmatched element, we put its two columns in a cluster. Now we have $\gamma N$ clusters with one free slot in each. One can fill in these slots by columns corresponding to $\gamma N/3$ unmatched triples. All the remaining unmatched triples form their own cluster. Now, consider an unmatched triple $s$ whose 3 columns are used to fill in slots of unmatched elements $p, q$ and $r$. As this triple was not matched, it cannot contain all these three elements, i.e, it can contain at most 2 of these elements. Thus for at least one element the cost of the cluster must be 9. Thus the total cost of the three clusters corresponding to $p, q$ and $r$ is at least $7+7+9=23$. The total cost corresponding to all $\gamma N/3$ unmatched triples is then $(23/3)\gamma N$. We will show that one cannot find a feasible clustering of lesser cost. 

For our convenience, we will prove the contrapositive of the second part of the above lemma: if there is a feasible clustering of cost less than  $7(1-\gamma)N+(23/3)\gamma N$, then there is a matching of size greater than $(1-\gamma)n$. So, assume that there is such a clustering. Let $c_1,c_2,\ldots,c_k$ be the cluster centers.  

By Lemma \ref{lem:exchange-old}, we can assume that if a column $f$ of $A$ is a center of a cluster $C$, all  the columns identical to $f$ are in $C$. We will use this in the following. A center $c_i$ is called an element center if $c_i$ is an element column. Suppose the given clustering contains $\ell$ clusters with element centers for some $\ell$. WLOG, we assume that these are the first $\ell$ clusters. 

\begin{lemma}\label{lem:bound-on-l}
If the cost of the given clustering is less than $7(1-\gamma)N+(23/3)\gamma N$, $\ell > (1-2\gamma/9)N$.  
\end{lemma}

\begin{proof}
Note that if a cluster center is an element column, then by Lemma \ref{lem:exchange-old} we can assume that both element columns are present in the cluster. Thus in our case, each of the first $\ell$ clusters contains two element columns and some other column. Now, each of these $\ell$ other columns can be either a column of some other element or a triple column. Let $\ell_1$ of these be element columns and $\ell_2$ of these be triple columns, where $\ell=\ell_1+\ell_2$. For each cluster corresponding to these $\ell_1$ element columns, the cost is 12, as $d_H(a_i,a_j)=12$ for all $i,j$. Similarly, for each cluster corresponding to the $\ell_2$ triple columns, the cost is at least 7, as $d_H(a_i,b_t)\ge 7$ for all $i,t$. 

Note that out of $2N$ element columns $2\ell+\ell_1$ are in the first $\ell$ clusters. The rest of the element columns are in the other clusters. Now there can be two cases: such a column is in a cluster that contains (i) at least 2 element columns and (ii) exactly one element column. 

\begin{claim}\label{cl:2-el-cluster}
The cost of each element column which are not in the first $\ell$ clusters is at least 5 in the first case.
\end{claim} 

\begin{proof}
Consider such a column $a_i$ and let $c_j$ be the center of the cluster that contains $a_i$. Note that the only 1 entries in $a_i$ are corresponding to the indexes in $R_i$. We claim that at most one entry of $c_j$ corresponding to the indexes in $R_i$ can be 1. This proves the original claim, as $|R_i|=6$. Consider an index $z \in R_i$ such that $c_j[z]=1$. As $c_j$ is not an element column and the centers are defined based on majority rule, there is a column $e$ in the cluster with $e[z]=1$. This must be a column of a triple that contains the element $i$. By construction, $e$ does not contain 1 corresponding to the indexes in $R_i \setminus \{z\}$. As the third column in the cluster is another element column (as we are in the first case), its entries corresponding to the indexes in $R_i$ are again 0. Hence, by majority rule, at most one entry of $c_j$ corresponding to the indexes in $R_i$ can be 1.         
\end{proof}

Next, we consider case (ii). 

\begin{claim}\label{cl:1-el-cluster}
Consider a cluster which is not one of the first $\ell$ clusters and contains exactly one element column. Then, its cost is at least 5. Moreover, the cost of the element column is at least 4. 
\end{claim}

\begin{proof}
Consider the element column $a_i$ of the cluster and let $c_j$ be the center of the cluster. Note that the only 1 entries in $a_i$ are corresponding to the indexes in $R_i$. Now, if the other two (triple) columns in the cluster are identical, there must be at most one entry of them corresponding to the indexes in $R_i$ that is 1. This is true by construction of triple columns. Hence, in this case at most one entry of $c_j$ corresponding to the indexes in $R_i$ can be 1 and the cost is at least 5. Otherwise, there can be two distinct triple columns $b_s$ and $b_t$ in the cluster and at most two indexes $z_1,z_2 \in R_i$ such that $z_1\ne z_2$ and $b_s[z_1]=b_t[z_2]=1$. By construction of the triple columns, there is no other indexes $z \in R_i \setminus \{z_1,z_2\}$ such that $b_s[z]=1$ or $b_t[z]=1$. Thus, by majority rule,  at most two entries of $c_j$ corresponding to the indexes in $R_i$ can be 1. Hence, the cost of $a_i$ is at least 4. Now, as $b_s$ and $b_t$ are distinct, the cost of either one of them must be at least 1. It follows that the cost of this cluster is at least 5. 
\end{proof}

Now, again consider the $2N-2\ell-\ell_1$ element columns that are not in the first $\ell$ clusters. Let $\kappa$ be the number of clusters which are not the first $\ell$ clusters and contain exactly 1 element column. This implies that, $2N-2\ell-\ell_1-\kappa$ element columns are contained in the clusters which are not the first $\ell$ clusters and contain at least 2 element columns. By, Claim \ref{cl:2-el-cluster}, the cost of each such column is at least 5. By Claim \ref{cl:1-el-cluster}, the cost of each of the $\kappa$ clusters defined above is at least 5. 

It follows that the total cost of the clustering is $12\ell_1+7\ell_2+(2N-2\ell-\ell_1-\kappa)5+5\kappa=10N-3\ell$, as $\ell=\ell_1+\ell_2$. Now, given that the cost is less than  $7(1-\gamma)N+(23/3)\gamma N$. 

\begin{align*}
& 10N-3\ell < 7(1-\gamma)N+(23/3)\gamma N=7N+2\gamma N/3\\    
& 3N-3\ell < 2\gamma N/3\\
& \ell > (1-2\gamma/9)N
\end{align*}
\end{proof}

Like before, let $\ell_2$ be the number of clusters out of the first $\ell$ clusters such that $\ell_2$ contains a triple column. 

\begin{claim}\label{cl;bound-on-l2}
$\ell_2 > (1-2\gamma/3)N$. 
\end{claim}

\begin{proof}
Again consider the cost of the given clustering. The cost of the $\ell_2$ clusters is at least 7. The cost of the remaining $\ell-\ell_2$ clusters is exactly 12 as before. Now, as $\ell > (1-2\gamma/9)N$ by Lemma \ref{lem:bound-on-l}, 
\begin{align*}
    & 7\ell_2+12((1-2\gamma/9)N-\ell_2) < 7(1-\gamma)N+(23/3)\gamma N=7N+2\gamma N/3\\    
    & 7\ell_2+12N-24\gamma N/9-12\ell_2 < 7N +2\gamma N/3\\
    & 5\ell_2 > 5N - 30\gamma N/9\\
    & \ell_2 > (1-2\gamma/3)N
\end{align*}
\end{proof}

We show that out of the $\ell_2$ elements corresponding to these $\ell_2$ clusters, more than $(1-\gamma)N$ elements must be matched. 

\begin{figure}[t]
\centering
\includegraphics[width=.9\linewidth]{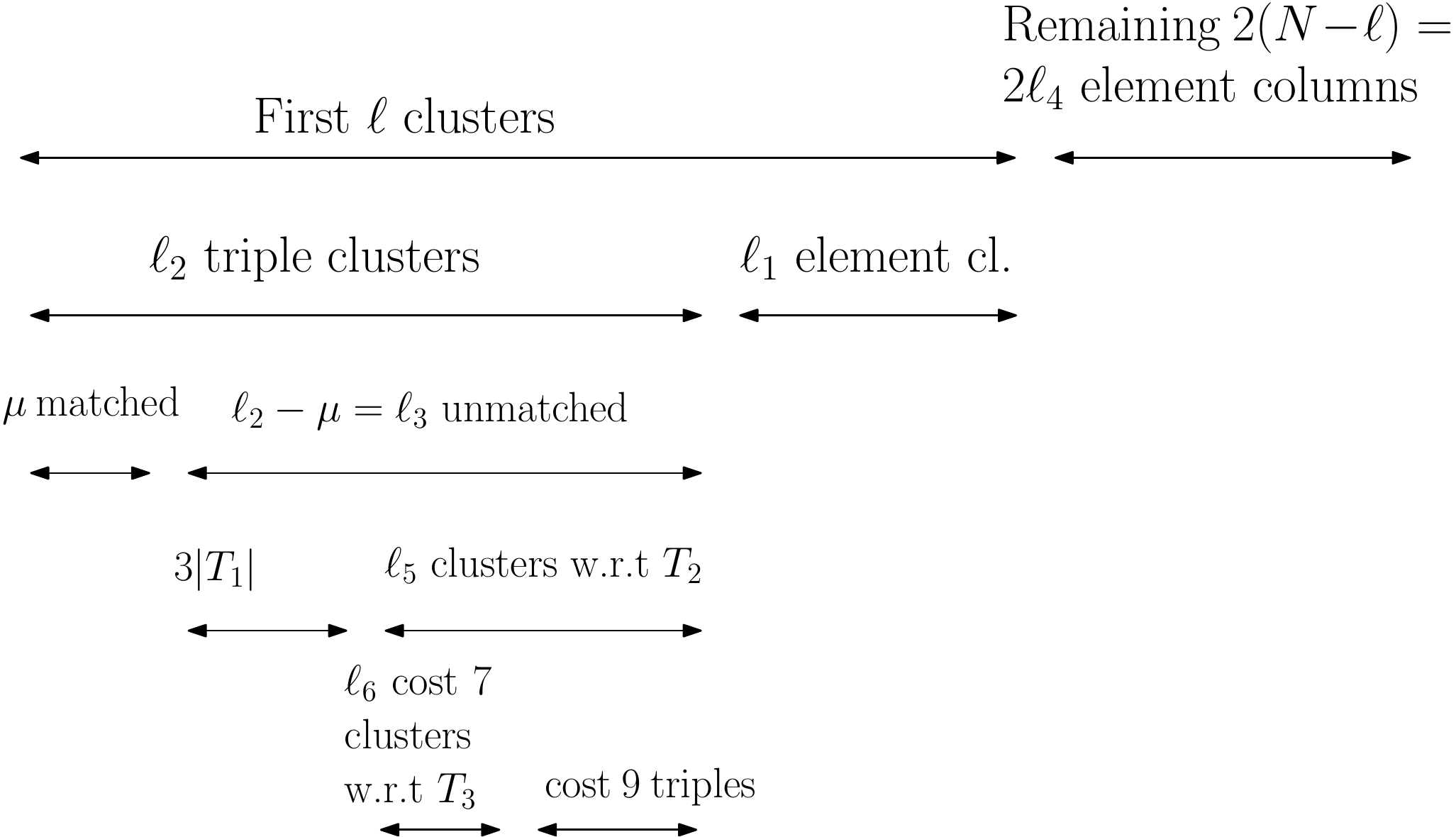}
\caption{Hierarchy of the clusters. Illustration of the proof of Lemma \ref{lem:matching}.}
\label{fig:hierarchy}
\end{figure}

\begin{lemma}\label{lem:matching}
There is a matching that matches more than $(1-\gamma)N$ elements. 
\end{lemma}

\begin{proof}
Consider the set of elements corresponding to the $\ell_2$ clusters, each of which contains a triple column. Let $M$ be a maximum matching involving these elements and triples that matches $\mu$ elements. We will show that $\mu > (1-\gamma)N$. The total cost of the clusters corresponding to these matched elements is $7\mu$. Let $\ell_{1}$ be the number of clusters out of the first $\ell$ clusters that contain all element columns (see Figure \ref{fig:hierarchy}). The total cost of these clusters is $12\ell_{1}$. Note that $3{\ell_1}$ columns are involved in these clusters. For the remaining at least $2(N-\mu) - 3{\ell_1}$ element columns and correspondingly at least $N-\mu - 3{\ell_1}/2$ elements, the corresponding columns can either be in one cluster along with a triple column or split into two clusters. Let $\ell_3$ be the number of such elements whose columns are in one cluster along with a triple column. Also, let $\ell_4$ be the remaining elements whose columns are split into two clusters (see Figure \ref{fig:hierarchy}). By Claims \ref{cl:1-el-cluster} and \ref{cl:2-el-cluster}, the cost of each split column is at least 4. Thus, the total cost corresponding to these $\ell_4$ elements is at least $8\ell_4$. Now, we compute the cost corresponding to the $\ell_3$ elements whose columns are in one cluster along with a triple column. Consider the set of triples involved in these clusters. Also, let $T_1$ be the set of triples whose three columns appear in these $\ell_3$ clusters. The cost of such triple columns is at least $7+7+9=23$, as they are not a part of the maximum matching. Let $\ell_5$ be the number of clusters among the $\ell_3$ clusters where the triples in $T_1$ does not appear and $T_2$ be the set of associated triples. Each triple in $T_2$ thus appear in at most 2 clusters among the $\ell_3$ clusters (see Figure \ref{fig:hierarchy}). Let $T_3\subseteq T_2$ be the set of triples each of which are only associated with the clusters of cost 7 and $\ell_6$ be the number of these clusters. As these triples are not part of the maximum matching, each of them can cover at most two unmatched elements. Thus, the size of $T_3$ is at least $\ell_6/2$. Note that, by definition, at least one column of each such triple does not belong to the first $\ell$ clusters. We compute the cost of these triple columns. If such a triple column appears in all triple column cluster, the cost of the column is at least 3, by construction of the triple columns and noting that two copies of the column cannot appear in the cluster. If such a triple is in a cluster with only one element column, its cost must be at least 2, as the element columns' at most one 1 entry can coincide with the 1 entries of the column. Now, if such a triple column appears in a cluster with two element columns, then the cost of the column is at least 1. However, the cost of the element columns must be at least 10. We charged each such element column a cost of 4 while charging the split columns corresponding to the $\ell_4$ elements. So, we can charge $10-8=2$ additional cost to those element columns. Instead, we charge this to the triple column. Thus, its charged cost is $1+2=3$. Thus, the total cost corresponding to the triples in $T_3$ is at least $(\ell_6/2)\cdot 2$. 

The total cost of the clustering is at least,
\begin{align*}
     & 7\mu+12\ell_1+8\ell_4+(23/3)|T_1|+(\ell_5-\ell_6)((7+9)/2)+7\ell_6+(\ell_6/2)\cdot 2\\
=  & 7\mu+12\ell_1+8\ell_4+(23/3)(\ell_3-\ell_5)+8\ell_5 \qquad\qquad\text{ (as } 3|T_1|=\ell_3-\ell_5)\\
\ge     &  7\mu+12\ell_1+8\ell_4+(23/3)\ell_3 \\
\ge & 7\mu+12\ell_1+(23/3)(\ell_3+\ell_4)\\
\ge & 7\mu+12\ell_1+(23/3)(N-\mu-3{\ell_1}/2) \qquad\qquad\text{ (as } \ell_3+\ell_4\ge N-\mu-3\ell_1/2)\\
= & 7\mu + (23/3)(N-\mu)+{\ell_1}/2\\
\ge & (23/3)N-(2/3)\mu \qquad\qquad\text{ (as } \ell_1\ge 0)\\
\end{align*}

Now, we know a strict upper bound on this cost. Thus, 
\begin{align*}
    & (23/3)N-(2/3)\mu < 7N+(2/3)\gamma N\\
    & (23/3-7)N-(2/3)\gamma N < (2/3)\mu\\
    & (2/3)N(1-\gamma) < (2/3)\mu\\
    & \mu > (1-\gamma)N
\end{align*}
\end{proof}

We summarize the results of this section in the following theorem. 

\begin{theorem}\label{thm:no-PTAS}   
There exists a constant $\epsilon_c > 0$, such that it is \classNP-hard to obtain a $(1+\epsilon_c)$-approximation for \probMEClust with $\ell_0$ $($or $\ell_1)$ distances, even if the input points are binary, that is, are from $\{0,1\}^d$. 
\end{theorem} 

\section{Conclusion}\label{sec:concl}
We initiated the study of lossy kernelization for clustering problems and proved that \probPEClust admits a $2$-approximation kernel. It is natural to ask whether the approximation factor may be improved. In particular, does the problem admit a \emph{polynomial size approximate kernelization scheme} (PSAKS) that is a lossy kernelization analog of PTAS (we refer to~\cite{FominLSZ19}  for the definition)?  Note that we proved that \probMEClust is {\sf APX}-hard and this refutes the existence of PTAS and makes it natural to ask the question about PSAKS. We also believe that it is interesting to consider the variants of the considered problems for means instead of medians. Here, the cost of a collection of points   $\bfX\subseteq \mathbb{Z}^d$ is defined as  
$\min_{\bfc\in\mathbb{R}^d}\sum_{\bfx\in \bfX}\|\bfc-\bfx\|_p^p$ for $p\geq 1$. Clearly, if $p=1$, that is, in the case of Manhattan norm, our results hold. However, for $p\geq 2$, we cannot translate our results directly, because our arguments rely on the triangle inequality. We would like to conclude the paper by underlining our belief that lossy kernelization may be natural tool for the lucrative area of approximation algorithms for clustering problems.


\begin{thebibliography}{10}

\bibitem{Aggarwalbook13}
{\sc C.~C. Aggarwal and C.~K. Reddy}, eds., {\em Data Clustering: Algorithms
  and Applications}, {CRC} Press, 2013.

\bibitem{DBLP:journals/corr/abs-1708-00622}
{\sc A.~Agrawal, S.~Saurabh, and P.~Tale}, {\em On the parameterized complexity
  of contraction to generalization of trees}, in Proceedings of the 12th
  International Symposium on Parameterized and Exact Computation (IPEC),
  vol.~89 of Leibniz International Proceedings in Informatics (LIPIcs), Schloss
  Dagstuhl - Leibniz-Zentrum fuer Informatik, 2017, pp.~1:1--1:12.

\bibitem{AlonS99}
{\sc N.~Alon and B.~Sudakov}, {\em On two segmentation problems}, Journal of
  Algorithms, 33 (1999), pp.~173--184.

\bibitem{baker2019coresets}
{\sc D.~Baker, V.~Braverman, L.~Huang, S.~H.-C. Jiang, R.~Krauthgamer, and
  X.~Wu}, {\em Coresets for clustering in graphs of bounded treewidth}, arXiv
  preprint arXiv:1907.04733,  (2019).

\bibitem{DBLP:journals/corr/abs-2007-10137}
{\sc S.~Bandyapadhyay, F.~V. Fomin, and K.~Simonov}, {\em On coresets for fair
  clustering in metric and euclidean spaces and their applications}, CoRR,
  abs/2007.10137 (2020).

\bibitem{basu2008constrained}
{\sc S.~Basu, I.~Davidson, and K.~Wagstaff}, {\em Constrained clustering:
  Advances in algorithms, theory, and applications}, CRC Press, 2008.

\bibitem{Bhattacharya2018}
{\sc A.~Bhattacharya, R.~Jaiswal, and A.~Kumar}, {\em {Faster Algorithms for
  the Constrained k-means Problem}}, Theory of Computing Systems, 62 (2018),
  pp.~93--115.

\bibitem{ByrkaFRS15}
{\sc J.~Byrka, K.~Fleszar, B.~Rybicki, and J.~Spoerhase}, {\em Bi-factor
  approximation algorithms for hard capacitated \emph{k}-median problems}, in
  Proceedings of the Twenty-Sixth Annual {ACM-SIAM} Symposium on Discrete
  Algorithms, {SODA} 2015, San Diego, CA, USA, January 4-6, 2015, P.~Indyk,
  ed., {SIAM}, 2015, pp.~722--736.

\bibitem{ByrkaRU16}
{\sc J.~Byrka, B.~Rybicki, and S.~Uniyal}, {\em An approximation algorithm for
  uniform capacitated k-median problem with $1+\varepsilon$ capacity
  violation}, in Integer Programming and Combinatorial Optimization - 18th
  International Conference, {IPCO} 2016, Li{\`{e}}ge, Belgium, June 1-3, 2016,
  Proceedings, Q.~Louveaux and M.~Skutella, eds., vol.~9682 of Lecture Notes in
  Computer Science, Springer, 2016, pp.~262--274.

\bibitem{CharikarGTS02}
{\sc M.~Charikar, S.~Guha, {\'{E}}.~Tardos, and D.~B. Shmoys}, {\em A
  constant-factor approximation algorithm for the k-median problem}, J. Comput.
  Syst. Sci., 65 (2002), pp.~129--149.

\bibitem{ChuzhoyR05}
{\sc J.~Chuzhoy and Y.~Rabani}, {\em Approximating k-median with non-uniform
  capacities}, in Proceedings of the Sixteenth Annual {ACM-SIAM} Symposium on
  Discrete Algorithms, {SODA} 2005, Vancouver, British Columbia, Canada,
  January 23-25, 2005, {SIAM}, 2005, pp.~952--958.

\bibitem{Cohen-Addad20}
{\sc V.~Cohen{-}Addad}, {\em Approximation schemes for capacitated clustering
  in doubling metrics}, in Proceedings of the 2020 {ACM-SIAM} Symposium on
  Discrete Algorithms, {SODA} 2020, Salt Lake City, UT, USA, January 5-8, 2020,
  S.~Chawla, ed., {SIAM}, 2020, pp.~2241--2259.

\bibitem{Cohen-AddadG0LL19}
{\sc V.~Cohen{-}Addad, A.~Gupta, A.~Kumar, E.~Lee, and J.~Li}, {\em Tight {FPT}
  approximations for k-median and k-means}, in Proceedings of the 46th
  International Colloquium on Automata, Languages, and Programming (ICALP),
  vol.~132 of LIPIcs, Schloss Dagstuhl - Leibniz-Zentrum f{\"{u}}r Informatik,
  2019, pp.~42:1--42:14.

\bibitem{Cohen-AddadL19}
{\sc V.~Cohen{-}Addad and J.~Li}, {\em On the fixed-parameter tractability of
  capacitated clustering}, in 46th International Colloquium on Automata,
  Languages, and Programming, (ICALP), vol.~132 of LIPIcs, Schloss Dagstuhl -
  Leibniz-Zentrum f{\"{u}}r Informatik, 2019, pp.~41:1--41:14.

\bibitem{CyganFKLMPPS15}
{\sc M.~Cygan, F.~V. Fomin, L.~Kowalik, D.~Lokshtanov, D.~Marx, M.~Pilipczuk,
  M.~Pilipczuk, and S.~Saurabh}, {\em Parameterized Algorithms}, Springer,
  2015.

\bibitem{DellM18}
{\sc H.~Dell and D.~Marx}, {\em Kernelization of packing problems}, CoRR,
  abs/1812.03155 (2018).

\bibitem{DemirciL16}
{\sc H.~G. Demirci and S.~Li}, {\em Constant approximation for capacitated
  k-median with $(1+\varepsilon)$-capacity violation}, in 43rd International
  Colloquium on Automata, Languages, and Programming (ICALP), 2016,
  pp.~73:1--73:14.

\bibitem{ding2020unified}
{\sc H.~Ding and J.~Xu}, {\em A unified framework for clustering constrained
  data without locality property}, Algorithmica, 82 (2020), pp.~808--852.

\bibitem{DowneyF13}
{\sc R.~G. Downey and M.~R. Fellows}, {\em Fundamentals of Parameterized
  Complexity}, Texts in Computer Science, Springer, 2013.

\bibitem{DBLP:journals/corr/newlossy}
{\sc P.~Dvo{\v r}{\'a}k, A.~E. Feldmann, D.~Knop, T.~Masa{\v r}{\'\i}k,
  T.~Toufar, and P.~Vesel{\'y}}, {\em Parameterized approximation schemes for
  steiner trees with small number of steiner vertices}, CoRR, abs/1710.00668
  (2017).

\bibitem{DBLP:journals/corr/EibenKMP17}
{\sc E.~Eiben, M.~Kumar, A.~E. Mouawad, F.~Panolan, and S.~Siebertz}, {\em
  Lossy kernels for connected dominating set on sparse graphs}, in Proceedings
  of the 34th International Symposium on Theoretical Aspects of Computer
  Science (STACS), vol.~96 of Leibniz International Proceedings in Informatics
  (LIPIcs), Schloss Dagstuhl - Leibniz-Zentrum fuer Informatik, 2018,
  pp.~29:1--29:15.

\bibitem{feldman2011unified}
{\sc D.~Feldman and M.~Langberg}, {\em A unified framework for approximating
  and clustering data}, in Proceedings of the 43rd Annual ACM Symposium on
  Theory of Computing (STOC), ACM, 2011, pp.~569--578.

\bibitem{feldman2013turning}
{\sc D.~Feldman, M.~Schmidt, and C.~Sohler}, {\em Turning big data into tiny
  data: Constant-size coresets for $k$-means, {PCA} and projective clustering},
  in Proceedings of the 23rd Annual ACM-SIAM Symposium on Discrete Algorithms
  (SODA), SIAM, 2013, pp.~1434--1453.

\bibitem{DBLP:journals/algorithms/FeldmannSLM20}
{\sc A.~E. Feldmann, {Karthik {C. S.}}, E.~Lee, and P.~Manurangsi}, {\em A
  survey on approximation in parameterized complexity: Hardness and
  algorithms}, Algorithms, 13 (2020), p.~146.

\bibitem{FominGP20}
{\sc F.~V. Fomin, P.~A. Golovach, and F.~Panolan}, {\em Parameterized low-rank
  binary matrix approximation}, Data Min. Knowl. Discov., 34 (2020),
  pp.~478--532.

\bibitem{FominGPS21}
{\sc F.~V. Fomin, P.~A. Golovach, and N.~Purohit}, {\em Parameterized
  complexity of categorical clustering with size constraints}, CoRR,
  abs/2104.07974 (2021).

\bibitem{FominLSZ19}
{\sc F.~V. Fomin, D.~Lokshtanov, S.~Saurabh, and M.~Zehavi}, {\em
  Kernelization}, Cambridge University Press, Cambridge, 2019.
\newblock Theory of parameterized preprocessing.

\bibitem{har2004coresets}
{\sc S.~Har-Peled and S.~Mazumdar}, {\em On coresets for $k$-means and
  $k$-median clustering}, in Proceedings of the 36th Annual ACM Symposium on
  Theory of Computing (STOC), ACM, 2004, pp.~291--300.

\bibitem{HoppnerK08}
{\sc F.~H{\"{o}}ppner and F.~Klawonn}, {\em Clustering with size constraints},
  in Computational Intelligence Paradigms, Innovative Applications, L.~C. Jain,
  M.~Sato{-}Ilic, M.~Virvou, G.~A. Tsihrintzis, V.~E. Balas, and C.~Abeynayake,
  eds., vol.~137, Springer, 2008, pp.~167--180.

\bibitem{KleinbergT06}
{\sc J.~M. Kleinberg and {\'{E}}.~Tardos}, {\em Algorithm design},
  Addison-Wesley, 2006.

\bibitem{DBLP:conf/fsttcs/KrithikaM0T16}
{\sc R.~Krithika, P.~Misra, A.~Rai, and P.~Tale}, {\em Lossy kernels for graph
  contraction problems}, in Proceedings of the 36th {IARCS} Annual Conference
  on Foundations of Software Technology and Theoretical Computer Science
  (FSTTCS), vol.~65 of Leibniz International Proceedings in Informatics
  (LIPIcs), Schloss Dagstuhl - Leibniz-Zentrum fuer Informatik, 2016,
  pp.~23:1--23:14.

\bibitem{Kuhn55}
{\sc H.~W. Kuhn}, {\em The {H}ungarian method for the assignment problem},
  Naval Res. Logist. Quart., 2 (1955), pp.~83--97.

\bibitem{DBLP:journals/jacm/KumarSS10}
{\sc A.~Kumar, Y.~Sabharwal, and S.~Sen}, {\em Linear-time approximation
  schemes for clustering problems in any dimensions}, J. {ACM}, 57 (2010),
  pp.~5:1--5:32.

\bibitem{Li15}
{\sc S.~Li}, {\em On uniform capacitated \emph{k}-median beyond the natural
  {LP} relaxation}, in Proceedings of the 26th Annual {ACM-SIAM} Symposium on
  Discrete Algorithms (SODA), 2015, pp.~696--707.

\bibitem{Li17}
\leavevmode\vrule height 2pt depth -1.6pt width 23pt, {\em On uniform
  capacitated \emph{k}-median beyond the natural {LP} relaxation}, {ACM} Trans.
  Algorithms, 13 (2017), pp.~22:1--22:18.

\bibitem{DBLP:conf/stoc/LokshtanovPRS17}
{\sc D.~Lokshtanov, F.~Panolan, M.~S. Ramanujan, and S.~Saurabh}, {\em Lossy
  kernelization}, in Proceedings of the 49th Annual ACM Symposium on Theory of
  Computing (STOC), ACM, 2017, pp.~224--237.

\bibitem{LovaszP09}
{\sc L.~Lov\'{a}sz and M.~D. Plummer}, {\em Matching theory}, AMS Chelsea
  Publishing, Providence, RI, 2009.

\bibitem{Patrank94}
{\sc E.~Petrank}, {\em The hardness of approximation: Gap location}, Comput.
  Complex., 4 (1994), pp.~133--157.

\bibitem{DBLP:journals/corr/Siebertz17a}
{\sc S.~Siebertz}, {\em Lossy kernels for connected distance-{$r$} domination
  on nowhere dense graph classes}, CoRR, abs/1707.09819 (2017).

\bibitem{sohler2018strong}
{\sc C.~Sohler and D.~P. Woodruff}, {\em Strong coresets for $k$-median and
  subspace approximation: Goodbye dimension}, in Proceedings of the 59th Annual
  Symposium on Foundations of Computer Science (FOCS), IEEE, 2018,
  pp.~802--813.

\bibitem{DBLP:conf/iccS/Vallejo-HuangaM17}
{\sc D.~Vallejo{-}Huanga, P.~Morillo, and C.~Ferri}, {\em Semi-supervised
  clustering algorithms for grouping scientific articles}, in International
  Conference on Computational Science (ICCS), vol.~108 of Procedia Computer
  Science, Elsevier, 2017, pp.~325--334.

\end{thebibliography}
\end{document}